\documentclass[journal,letterpaper,twocolumn]{IEEEtran}
\usepackage{amsfonts}
\usepackage{amssymb}
\usepackage{cite}
\usepackage[cmex10]{amsmath}
\usepackage{float}
\usepackage{color}
\usepackage{stfloats,fancyhdr}

\usepackage{algorithm}
\usepackage{algorithmic}
\usepackage{multirow}
\usepackage{changepage}
\usepackage[normalem]{ulem}
\usepackage{flushend,cuted}

\newtheorem{theorem}{Theorem}
\newtheorem{lemma}{Lemma}
\newtheorem{proposition}{Proposition}
\newtheorem{corollary}{Corollary}
\newtheorem{definition}{Definition}

\IEEEoverridecommandlockouts

\ifCLASSINFOpdf
  \usepackage[pdftex]{graphicx}
  \DeclareGraphicsExtensions{.pdf,.jpeg,.png}
\else
  \usepackage[dvips]{graphicx}
  \DeclareGraphicsExtensions{.eps}
\fi

\usepackage{subfigure}

\usepackage{fancybox,dashbox}
\begin{document}

%: A Game-Theoretical Approach
\title{{Distributed Power Splitting for SWIPT in Relay Interference Channels using Game~Theory%: A Game-Theoretical Approach
}
\thanks{The work was supported by the Australian Research Council (ARC) under
Grants DP120100190 and FT120100487, International Postgraduate Research Scholarship (IPRS), Australian Postgraduate Award (APA), and Norman I Price Supplementary Scholarship.}
\thanks{Part of this work was presented at the IEEE International Symposium on Information Theory (ISIT), Honolulu, HI, USA, July 2014 \cite{Chen_ISIT_2014_A}.}
\thanks{H. Chen, Y. Li, Y. Ma, and B. Vucetic are with the School of Electrical and Information Engineering, The University of Sydney, Sydney, NSW 2006, Australia (email: he.chen@sydney.edu.au, yonghui.li@sydney.edu.au, yuanye.ma@sydney.edu.au, branka.vucetic@sydney.edu.au).}
\thanks{Y. Jiang is with Department of Electronic and Information Engineering, Hong Kong Polytechnic University, Hong Kong, China (email:yunxiang.jiang@connect.polyu.hk).}
}

\author{He~(Henry)~Chen,~\IEEEmembership{Student Member,~IEEE,}
Yonghui Li,~\IEEEmembership{Senior Member,~IEEE,} \\Yunxiang Jiang,
Yuanye Ma, and Branka Vucetic,~\IEEEmembership{Fellow,~IEEE}
}

\maketitle

\begin{abstract}
In this paper, we consider simultaneous wireless information and power transfer (SWIPT) in relay interference channels, where multiple source-destination pairs communicate through their dedicated energy harvesting relays. Each relay needs to split its received signal from sources into two streams: one for information forwarding and the other for energy harvesting. We develop a distributed power splitting framework using game theory to derive a profile of power splitting ratios for all relays that can achieve a good network-wide performance. Specifically, non-cooperative games are respectively formulated for pure amplify-and-forward (AF) and decode-and-forward (DF) networks, in which each link is modeled as a strategic player who aims to maximize its own achievable rate. The existence and uniqueness for the Nash equilibriums (NEs) of the formulated games are analyzed and a distributed algorithm with provable convergence to achieve the NEs is also developed. Subsequently, the developed framework is extended to the more general network setting with mixed AF and DF relays. All the theoretical analyses are validated by extensive numerical results. Simulation results show that the proposed game-theoretical approach can achieve a near-optimal network-wide performance on average, especially for the scenarios with relatively low and moderate interference.
\end{abstract}

\IEEEpeerreviewmaketitle

\section{Introduction}
Harvesting energy from the environment has been regarded as a promising technique to prolong the lifetime of energy constrained wireless networks, e.g., sensor networks. Apart from the conventional renewable energy sources such as solar and wind, radio frequency (RF) signals radiated by ambient transmitters can be treated as a viable new source for wireless power transfer. In addition, RF signals have been widely used for wireless information transmission. Therefore, simultaneous wireless information and power transfer (SWIPT) \cite{Varshney_ISIT_2008} has become a promising approach to enable
%becomes attractive since it is able to realize the dual utilization of the same RF signals for
new ways of information and energy delivery. The basic concept of SWIPT was first proposed in \cite{Varshney_ISIT_2008}. A comprehensive receiver architecture and the corresponding rate-energy tradeoff were developed in \cite{Zhou_J_2012}. Inspired by these two seminal results, SWIPT schemes for wireless networks have attracted a lot of attention recently. SWIPT schemes for MIMO broadcasting channels and MISO interference channels were designed and evaluated in \cite{Zhang_TWC_2013_MIMO} and \cite{Timotheou_TWC_2013}, respectively. The resource allocation algorithms for SWIPT in broadband wireless systems were investigated in \cite{Huang_TSP_2013}, while an energy-efficient resource allocation algorithm was developed in \cite{Ng_TWC_2013} for SWIPT in OFDMA systems.

Recently, Nasir \emph{et al.} extended the idea of SWIPT to a simple three-node relay network in \cite{Zhou_TWC_2013}, where the energy constrained relay node harvests energy from the source signal to enable forwarding the received signal. Two practical relaying protocols were proposed for the considered relay network. Analytical expressions for the outage probability and the ergodic capacity of the proposed protocols were derived for delay-limited and delay-tolerant modes, respectively. This work was further extended in \cite{Nasir_Tcom_2013}, by considering an adaptive time-switching protocol for SWIPT relaying network and analytical expressions of the achievable throughput were derived for both amplify-and-forward and decode-and-forward relaying networks. A similar idea for SWIPT in two-way relaying networks was proposed and analyzed in \cite{Chen_ICC_2014}. Very recently, \cite{Ding_arXiv_2013} considered SWIPT in a larger relay network, where multiple source-destination pairs communicate with each other via a common energy harvesting relay. Specifically, several power allocation schemes were proposed to efficiently distribute the power harvested at the relay among multiple pairs.

In cooperative networks, the relay interference channel is an important model that has many practical applications, such as cellular networks, wireless sensor networks, WLANs etc. \cite{Simeone_Conf_2007,Thejaswi_Conf_2008,Cao_Conf_2009, Shi_TWC_2009,Yi_TWC_2012,Truong_EJASP_2013}. In this model, multiple source-destination pairs communicate with the help of dedicated relays using the same spectral resources\footnote{Such a model is also called two-hop interference channels since the multiple source-relay-destination links, which constitute a cascade of two interference channels, transmit and interfere in each hop.}. However, to the best of our knowledge, no work in the open literature has considered SWIPT in relay interference channels except \cite{Krikidis_Tcom_2013}. Note that although \cite{Zhou_TWC_2013,Nasir_Tcom_2013,Chen_ICC_2014,Ding_arXiv_2013} also designed SWIPT for relay networks, none of them considered relay interference channels. By using the advanced stochastic geometry theory, \cite{Krikidis_Tcom_2013} analyzed the outage performance of SWIPT in large-scale networks with/without relaying, where the transmitters and the relays (if they exist) are assumed to be connected to a power supply, while the receivers harvest the energy from the signals received from the source and relay based on a power splitting technique \cite{Zhou_J_2012}.

In this paper, we also focus on the design of SWIPT in relay interference channels. Different from \cite{Krikidis_Tcom_2013} in the model, objective, and approach, we consider that multiple source-destination pairs communicate simultaneously with the help of their dedicated \emph{RF energy harvesting} relays, which do not have their own power supply and need to harvest energy from the source signal before forwarding. Each relay node splits the signal received from all source nodes into two parts according to a power splitting ratio: one part is sent to the information processing unit, and the rest is used to harvest energy for forwarding the received information in the second time slot. We consider that each link's performance is characterized by its achievable rate and thus regard the sum-rate of all links as a network-wide performance metric. The first natural question that arises from this system is ``\emph{how should the relays split their received signals for information receiving and energy harvesting in order to achieve a good network-wide performance?}". This is actually a very complex question to answer. The reason is that the power splitting ratio of each link not only affects the performance of this link, but also affects the performance of other links due to mutual interference between different links. This means that the optimization of each ratio depends on all other ratios and they are tangled together. Moreover, the maximization of the sum-rate of all links is shown to be a non-convex optimization problem. The global optimal power splitting ratios cannot be efficiently achieved even in a centralized fashion, and there is a heavy signaling overhead required by the centralized method.

To tackle the aforementioned problem, we apply the well-established game theory to develop a distributed power splitting framework for SWIPT in relay interference channels. We investigate both pure and hybrid networks in this paper. In a pure network, all relays adopt the same relaying protocol. Considering that amplify-and-forward (AF) relaying and decode-and-forward (DF) relaying protocols are most-frequently used in practice \cite{Yonghui_book_2010}, we further classify a pure network into a pure AF network and a pure DF network. On the other hand, in a hybrid network, a mixture of AF and DF relaying protocols are implemented at the relays. To the best of our knowledge, this is the first game-theoretical framework for the design of SWIPT in relay interference channels. The main contributions of the paper are summarized as follows:
\begin{itemize}
  \item We develop a distributed power splitting framework for the SWIPT in relay interference channels. In particular, each source-relay-destination link in the relay interference channels is modeled as a strategic player who chooses its dedicated relay's power splitting ratio to maximize its individual~rate.
  \item We analyze the existence and uniqueness of the Nash equilibrium (NE) for the formulated game in the pure network, where all relays employ either AF or DF relaying protocol. In addition, a distributed algorithm is proposed with provable convergence to achieve the NEs.
  \item The theoretical analysis for the pure networks is then extended to a more general hybrid network with mixed AF and DF relays coexisting.
  \item All analytical results are validated by extensive numerical simulations, which show that the proposed game-theoretical approach can achieve a near-optimal network-wide performance on average.
\end{itemize}

The remainder of this paper is organized as follows. Section II describes the system model. In Section III, we present the proposed game-theoretical power splitting framework for the pure networks, where the non-cooperative games are formulated for the pure AF and pure DF networks, followed by the existence and uniqueness analysis of the NEs as well as the development of the distributed algorithm. The extension of the proposed framework to the hybrid network is discussed in Section IV. Numerical results are provided in Section V and conclusions are drawn in Section VI.

\section{System Model}

We consider SWIPT in a system with relay interference channels, as depicted in Fig. \ref{fig1}. The system consists of $N$ source-relay-destination ($S$-$R$-$D$) links and the set of these links is denoted as ${\mathcal N} = \left\{ {1, \ldots ,N} \right\}$. More specifically, in the link ${S_i} \to {R_i} \to {D_i}$, $i \in {\mathcal N} $, the source $S_i$ communicates with its corresponding destination $D_i$, assisted by a dedicated relay $R_i$. The relay nodes can employ either AF or DF relaying schemes \cite{Nosratinia_CM_2004}. The direct source-destination channels are neglected due to a high path loss and shadowing attenuation. Since these two-hop links share the same spectrum, they interfere with each other over the dual hops.

\begin{figure}
\centering \scalebox{0.15}{\includegraphics{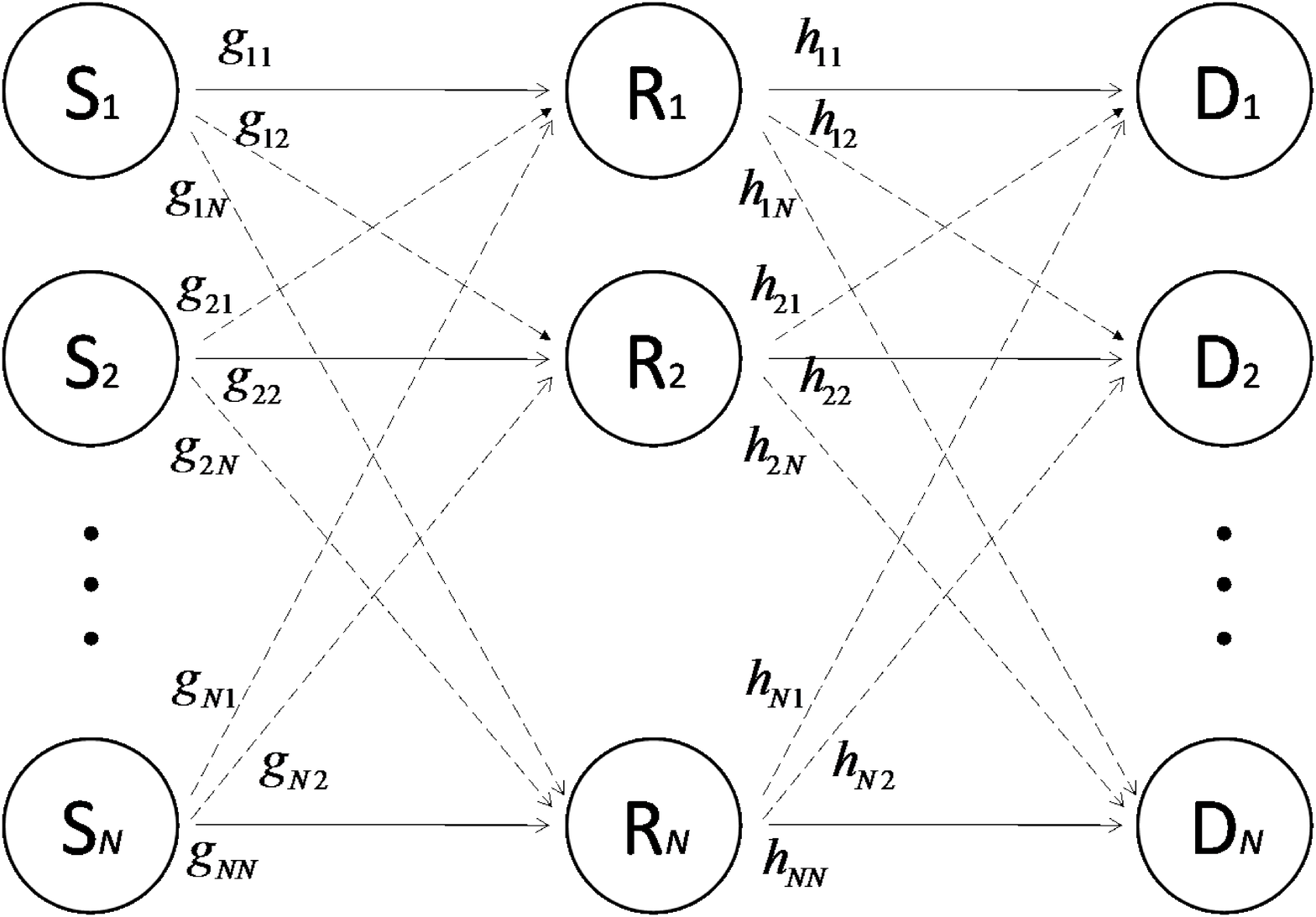}}
\caption{System model for interference relay channels. \label{fig1}}
\end{figure}

We assume that all the nodes (i.e., sources, relays and destinations) are equipped with only one antenna and operate in a half-duplex
mode. The relay nodes do not have their own power supply and need to harvest energy from the
received signal in order to forward the received signal to the destinations. It is assumed that the energy harvesting and information transmission are implemented for every received message block. For the purpose of exposition, \emph{the processing power consumed by the transmit/receive circuitry at the relay nodes is assumed to be negligible as compared to the power used for signal forwarding} \cite{Zhou_TWC_2013}. Moreover, we consider that all links experience slow and frequency-flat fading.

Let $g_{ij}$ and $h_{ij}$ denote the channel gain from $S_i$ to $R_j$ in the first hop and from $R_i$ to $D_j$ in the second hop, respectively. In the first time slot, all sources transmit simultaneously and the signal received by the relay $R_i$ can be written as
\begin{equation}\label{eq:received_signal_relay}
{y_{{R_i}}} = \sqrt {{P_{{i}}}} {g_{ii}}{x_i} + \sum\nolimits_{j = 1,j \ne i}^N {\sqrt {{P_{{j}}}} {g_{ji}}{x_j} }+ n_i^a,
\end{equation}
where $P_i$ and $x_i$ are the fixed transmit power and the transmitted information of the source $S_i$ with $\mathbb{E}\left\{ \left| x_i \right|^2 \right\} = 1$, $n_i^a \sim {\mathcal{CN}}\left(0,\sigma _{i,a}^2\right)$ is the additive noise introduced by the receiver antenna at the relay $R_i$.

Subsequently, the received signal at the relay $R_i$ is split into two streams, with the power splitting ratio $\rho_i$, as shown in Fig.~\ref{fig2}. The fraction $\sqrt{\rho_i}$ of the received signal is used for energy harvesting, while the remaining one is sent to the information processing unit. In practice, the antenna noise $n_i^a$ has a negligible impact on both the information receiving and energy harvesting, since $\sigma _{i,a}^2$ is generally much smaller than the noise power introduced by the baseband processing circuit, and thus even lower than the average power of the received signal \cite{Liu_Tcom_2013}. For simplicity, we ignore the noise term $n_i^a$ in the following analysis, i.e., setting $\sigma _{i,a}^2 = 0$ \cite{Liu_Tcom_2013}. For the sake of simplicity, we assume a normalized transmission time for each hop (i.e., the transmission duration of each hop is equal to one). Then, the terms ``energy" and ``power" can be used interchangeably. In this case, the energy harvested at relay $R_i$ can thus be expressed as
\begin{equation}\label{eq:energy_harvested_relay}
{Q_i} = \eta {\rho _i}%\left( {{P_i}{{\left| {{g_{ii}}} \right|}^2} + \sum\nolimits_{j = 1,j \ne i}^N { {P_j}{{\left| {{g_{ji}}} \right|}^2}}          } \right),
\sum\nolimits_{n = 1}^N { {P_n}{{\left| {{g_{ni}}} \right|}^2}},
\end{equation}
where $0< \eta \le 1$ is the energy conversion efficiency that depends on the rectification process and the energy harvesting circuit. Meanwhile, the information signal received by the information processing unit at relay $R_i$ is given by
\begin{equation}\label{eq:information_signal_relay}
\begin{split}
y_{{R_i}}^I =& \sqrt {1 - {\rho _i}} {y_{{R_i}}} + n_i^b\\
=& \sqrt {1 - {\rho _i}} \sqrt {{P_i}} {g_{ii}}{x_i} + \sqrt {1 - {\rho _i}}\times\\ &\sum\nolimits_{j = 1,j \ne i}^N {\sqrt {{P_j}} {g_{ji}}{x_j}}  + n_i^b,
\end{split}
\end{equation}
where $n_i^b \sim {\mathcal{CN}}\left(0,\sigma _{{R_i}}^2\right)$ is the additive white Gaussian noise introduced by the signal processing circuit from passband to baseband. Then, in the second time slot, the relay nodes will exhaust\footnote{Generally, the relay nodes may be interested in keeping part of the energy harvested from the RF signals. In this paper, we consider the relay protocol to maximize achievable rate of each link. In this regard, the relay should exhaust the harvested energy to forward the source information and thus has no incentive to keep any part of energy.} the harvested energy to forward the information signal $y_{{R_i}}^I$ by employing either AF or DF relaying protocol. In the following subsections, we derive the expressions for the achievable rates of the $i$th link for the AF and DF relaying protocols, respectively.
\begin{figure}
\centering \scalebox{0.35}{\includegraphics{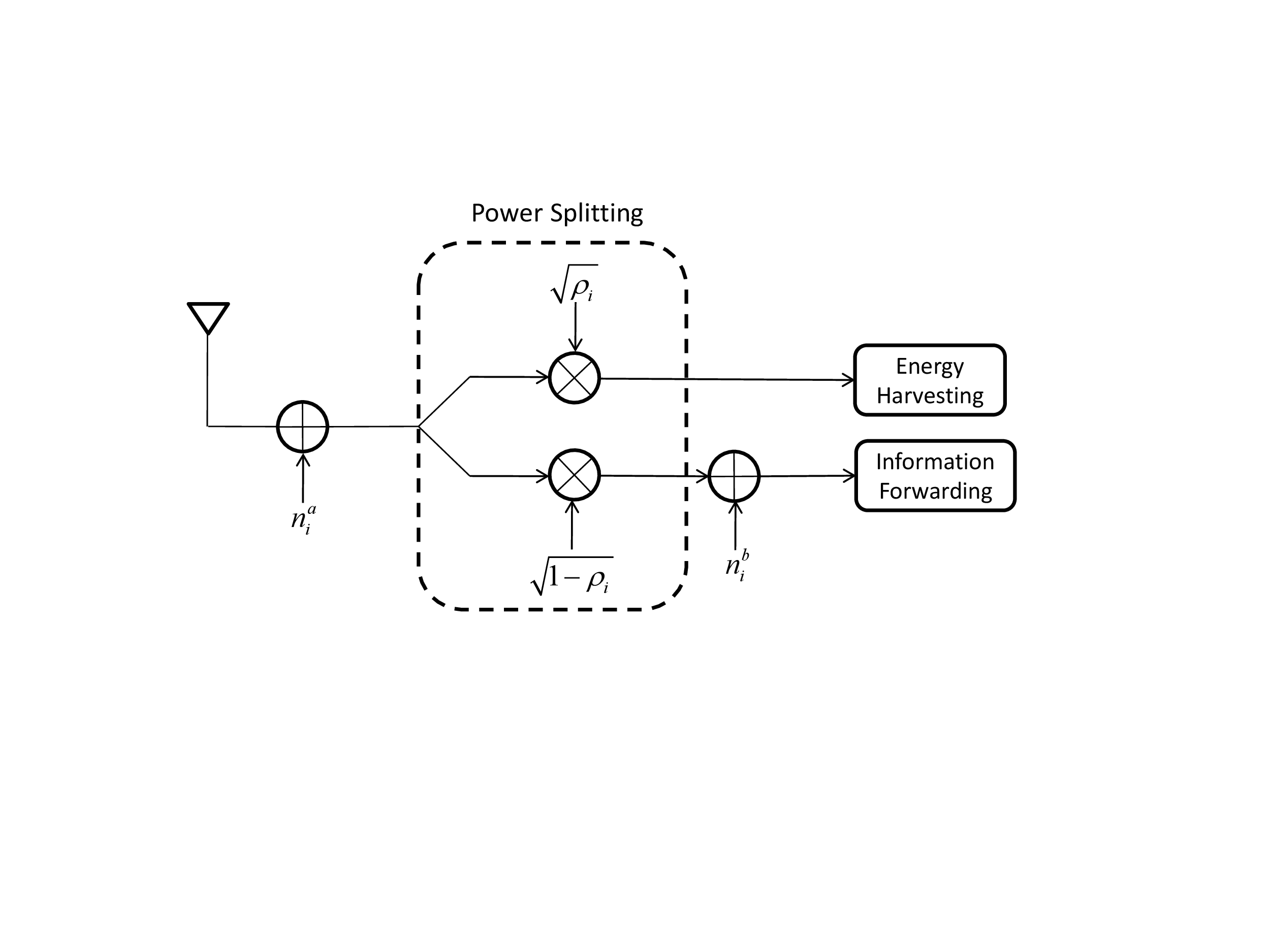}}
\caption{The diagram of the power splitting technique at the relay nodes. \label{fig2}}
\end{figure}

\subsection{AF Relaying}
When the AF relaying scheme is adopted, the relay node $R_i$ will exhaust the harvested energy to amplify and forward the signal received by the information processing unit in the first time slot. Thus, the transmit power of the relay $R_i$ is $Q_i$ and the received signal at the destination $D_i$ can be expressed as \cite{Yonghui_book_2010}
\begin{equation}\label{eq:destination_signal_AF}
\begin{split}
y_{{D_i}}^{AF} =& \sqrt {{Q_i}} {h_{ii}}{\beta _i}y_{{R_i}}^I + \sum\nolimits_{j = 1,j \ne i}^N {\sqrt {{Q_j}} {h_{ji}}{\beta _j}y_{{R_j}}^I}  + {n_{{D_i}}}\\
=& \sqrt {{Q_i}} {h_{ii}}{\beta _i}\sqrt {1 - {\rho _i}} \sqrt {{P_i}} {g_{ii}}{x_i} +\\ &\sqrt {{Q_i}} {h_{ii}}{\beta _i}\sqrt {1 - {\rho _i}} \sum\nolimits_{j = 1,j \ne i}^N {\sqrt {{P_j}} {g_{ji}}{x_j}}+\\
& \sqrt {{Q_i}} {h_{ii}}{\beta _i}n_i^b +\sum\nolimits_{j = 1,j \ne i}^N {\sqrt {{Q_j}} {h_{ji}}{\beta _j}y_{{R_j}}^I}  + {n_{{D_i}}},
\end{split}
\end{equation}
where
\[{\beta _i} = 1/\sqrt {\left( {1 - {\rho _i}} \right){P_i}{{\left| {{g_{ii}}} \right|}^2} + \left( {1 - {\rho _i}} \right)\sum\limits_{j = 1,j \ne i}^N {{P_j}{{\left| {{g_{ji}}} \right|}^2}}  + \sigma _{{R_i}}^2} \]
is the power constraint factor at the relay $R_i$, and ${n_{{D_i}}}\sim {\mathcal{CN}}\left(0,\sigma _{{D_i}}^2\right)$ is the additive white noise at the destination $D_i$. Without loss of generality, we hereafter assume that $\sigma _{{R_i}}^2 = \sigma _{{D_i}}^2 = \sigma ^2$, for any $i \in {{\mathcal N}}$. The second equality of (\ref{eq:destination_signal_AF}) is obtained by inserting the expression of $y_{{R_i}}^I$ given by (\ref{eq:information_signal_relay}) into the first equality of (\ref{eq:destination_signal_AF}). Note that only the first term on the right-hand side of the second equality in (\ref{eq:destination_signal_AF}) is the useful signal to the destination $D_i$, while the remaining terms should be regarded as interference plus noise. Based on this observation and after some algebraic manipulations, we can write the end-to-end signal-to-interference-plus-noise ratio (SINR) of the $i$th link as (\ref{eq:SINR_AF}) on top of next page,
\begin{figure*}
\begin{equation}\label{eq:SINR_AF}
\gamma _i^{AF} = \frac{{{\rho _i}\left( {1 - {\rho _i}} \right){X_i}{Z_i}}}{{{\rho _i}\left( {1 - {\rho _i}} \right){Y_i}{Z_i} + \left( {1 - {\rho _i}} \right)\left( {{X_i} + {Y_i}} \right)\left( {{W_i} + 1} \right) + {\rho _i}{Z_i} + {W_i} + 1}}%,
\end{equation}
\hrulefill \vspace*{4pt}
\end{figure*}
where
\begin{subequations}\label{eq:def_XYZW}
\begin{align}
&{X_i} = {P_i}{\left| {{g_{ii}}} \right|^2}/{\sigma ^2}, \label{eq:def_XYZW_a}\\
&{Y_i} = \sum\nolimits_{j = 1,j \ne i}^N {{P_j}{{\left| {{g_{ji}}} \right|}^2}} /{\sigma ^2},\label{eq:def_XYZW_b}\\
&{Z_i} = \eta \left( {\sum\nolimits_{n = 1}^N {{P_n}{{\left| {{g_{ni}}} \right|}^2}} } \right){\left| {{h_{ii}}} \right|^2}/{\sigma ^2},\label{eq:def_XYZW_c}\\
&{W_i} = \sum\nolimits_{j = 1,j \ne i}^N {{\rho _j}\eta \left( {\sum\nolimits_{n = 1}^N {{P_n}{{\left| {{g_{nj}}} \right|}^2}} } \right){{\left| {{h_{ji}}} \right|}^2}}/ {\sigma ^2},\label{eq:def_XYZW_d}
\end{align}
\end{subequations}
are defined for the simplicity of notations. It is worth noticing that the above equations (\ref{eq:def_XYZW}a)-(\ref{eq:def_XYZW}d) have physical meanings. More specifically, (\ref{eq:def_XYZW}a) and (\ref{eq:def_XYZW}b) respectively denote the signal-to-noise ratio (SNR) and the interference-to-noise ratio (INR) at the relay $R_i$ when the received signal is fully forwarded to $D_i$ without harvesting any energy at the relay (i.e., $\rho_i = 0$). On the other hand, (\ref{eq:def_XYZW}c) represents the SNR at the destination $D_i$ when the received signal of the relay $R_i$ is fully used for energy harvesting (i.e., $\rho_i = 1$). Finally, (\ref{eq:def_XYZW}d) is the INR at the destination $D_i$.

Then, the achievable rate of the link $i$ when the AF relaying technique is employed at its dedicated relay can be expressed as \begin{equation}\label{eq: utiliy_AF}
{u_i^{AF}}\left( \pmb \rho  \right) = \frac{1}{2}\log \left( {1 + \gamma _i^{AF}} \right),
\end{equation}
where ${\pmb \rho} = \left[ {{\rho _1}, \ldots ,{\rho _N}} \right]^T$ denotes the vector of all links' power splitting ratios.

\subsection{DF Relaying}
For the case when the DF relaying protocol is employed, the relay node will first decode the information based on the received information signal $y_{{R_i}}^I$ given in (\ref{eq:information_signal_relay}). Thus, the received SINR at relay $R_i$ can be written as
\begin{equation}\label{eq:DF_SNR_1st_hop}
\gamma _{i,1}^{DF} = \frac{{\left( {1 - {\rho _i}} \right)X_i}}{{\left( {1 - {\rho _i}} \right)Y_i + 1}},
\end{equation}
where $X_i$ and $Y_i$ are defined in (\ref{eq:def_XYZW_a}) and (\ref{eq:def_XYZW_b}), respectively.

In the second time slot, the relay nodes forward the decoded information to their corresponding destinations using the energy harvested in the first time slot. The received signal at $D_i$ is given by
\begin{equation}\label{eq:received_signal_destination}
{y_{{D_i}}} = \sqrt {{Q_i}} {h_{ii}}{{ x}_i} + \sum\nolimits_{j = 1,j \ne i}^N { \sqrt {{Q_j}} {h_{ji}}{{ x}_j} }+ {n_{{D_i}}}.
\end{equation}
The received SINR at the destination $D_i$ can thus be written as
\begin{equation}\label{eq:DF_SNR_2rd_hop}
\begin{split}
\gamma _{i,2}^{DF} = \frac{{ {\rho _i}Z_i}}{W_i + 1},
\end{split}
\end{equation}
where $Z_i$ and $W_i$ are defined in (\ref{eq:def_XYZW_c}) and (\ref{eq:def_XYZW_d}), respectively. The achievable rate of the $i$th link in this case is thus given by
\begin{equation}\label{eq:utility_DF}
\begin{split}
u_i^{DF}\left( \pmb \rho  \right) &=
\frac{1}{2}\min \left( {\log \left( {1 + \gamma _{i,1}^{DF}} \right),\log \left( {1 + \gamma _{i,2}^{DF}} \right)} \right)\\
&=\frac{1}{2}\log \left( {1 + \gamma _i^{DF}}\right) ,
\end{split}
\end{equation}
where
\begin{equation}\label{eq:SINR_e2e_DF}
\gamma _i^{DF} = \min \left( {\gamma _{i,1}^{DF},\gamma _{i,2}^{DF}} \right)
\end{equation}
can be regarded as the end-to-end SINR of the $i$th link with a DF relay.

In this paper, we consider that each link's performance is characterized by its achievable rate and thus regard the sum-rate of all links as a network-wide performance metric. In the following sections, we will develop a distributed power splitting scheme to achieve a good network-wide performance.

\section{Distributed Power Splitting for Pure Networks}\label{sec:Pure_AF_Network}
In this section, we focus on the the design of distributed power splitting for pure AF and DF networks, where all the relay nodes employ the same relaying protocol, i.e., either AF or DF relaying. To choose an efficient profile of the power splitting ratios (i.e, $\pmb \rho$) that can achieve a globally optimal network-wide performance, one needs to solve the following network utility maximization problem:
\begin{equation}\label{eq:centralized_optimization_AF}
\begin{array}{l}
 \mathop {\max }\limits_{\pmb \rho}  \sum\nolimits_{i = 1}^N {u_i^{X}\left( {\pmb \rho}  \right)}  \\
 ~~~{\rm {s.t.}}\;\;\rho  \in {\mathcal A} \\
\end{array},
\end{equation}
where $X$ refers to $AF$ ($DF$) for the pure AF (DF) network, ${{u_i^{AF}}\left( {\pmb \rho}  \right)}$ and ${{u_i^{DF}}\left( {\pmb \rho}  \right)}$ are respectively defined in (\ref{eq: utiliy_AF}) and (\ref{eq:utility_DF}), and ${\mathcal A} = \left\{ {\left. {\pmb \rho}  \right|0 \le {\rho _i} \le 1,\;\forall i \in {\mathcal N}} \right\}$ is the feasible set of ${\pmb \rho}$.

However, it can be easily checked that the optimization problem in (\ref{eq:centralized_optimization_AF}) is not convex for an AF network. Moreover, for a DF network, the optimization problem in (\ref{eq:centralized_optimization_AF}) is not only non-convex but also non-differentiable due to the $\min$ operator. This means that the globally optimal power splitting profile for the pure network (i.e, the solution of (\ref{eq:centralized_optimization_AF})) cannot be efficiently calculated even in a centralized fashion and there will be a heavy signaling overhead required by the centralized method. Motivated by this, we will develop a distributed framework by considering that all links are strategic and they aim to maximize their individual achievable rates by choosing their own power splitting ratios. For example, in an AF network, this will involve the $i$th link solving the following optimization problem
\begin{equation}\label{eq:distributed_optimization_AF}
\begin{array}{l}
\mathop {\max }\limits_{{\rho _i}} \;\;\;{u_i^{AF}}\left( {{\rho _i},{\pmb \rho _{ - i}}} \right)\\
~~~~{\rm{s.t.}}\;\;\;{\rho _i} \in {{\mathcal A}_i}
\end{array},
\end{equation}
where ${{\pmb \rho} _{ - i}} = {\left[ {{\rho _1}, \ldots ,{\rho _{i - 1}},{\rho _{i + 1}}, \ldots {\rho _N}} \right]^T}$ denotes the vector of all links' power splitting ratios, except the $i$th one, and ${\cal A}_i = \left\{ {\left. \rho_i  \right|0 \le {\rho _i} \le 1} \right\}$ is the feasible set of the $i$th link's power splitting ratio.

We can observe from (\ref{eq:distributed_optimization_AF}) that the optimization problem to be solved by each link is coupled together due to the mutual interference over two hops. To solve this problem, we model the considered power splitting problem to be a non-cooperative game in game theory \cite{Fudenberg_Game_book}.
Particularly, the considered power splitting problem for an AF network can be modeled by the following non-cooperative game:
\begin{itemize}
\item \emph{{Players}}: The $N$ $S$-$R$-$D$ links.
\item \emph{{Actions}}: Each link determines its power splitting ratio $\rho_i \in {\mathcal{A}}_i$ to maximize the achievable rate for its own link.
\item \emph{{Utilities}}: The achievable rate ${u_i^{AF}}\left( {{\rho _i},{\pmb \rho _{ - i}}} \right)$ defined in (\ref{eq: utiliy_AF}).
\end{itemize}

For convenience, we denote the formulated non-cooperative game as
\begin{equation}\label{eq:AF_game}
{\mathcal{G}}_{AF} = \left\langle {\mathcal{N},\left\{ \mathcal{A}_i  \right\},\left\{{{u}_i^{AF}}\left( {{\rho_i},{{{\pmb \rho}}_{ -i}}} \right)\right\} }\right\rangle.
\end{equation}
Note that we regard each link consisting of three nodes as a ``virtual" single player for the sake of presentation. In practice, each player is supposed to be one node of each link (e.g., relay) that acts as the coordinator of each link.

Similarly, we can formulate the following non-cooperative game for the DF network:
\begin{equation}\label{eq:game_DF}
{\mathcal{G}}_{DF} = \left\langle {\mathcal{N},\left\{ \mathcal{A}_i  \right\},\left\{{{u}_i^{DF}}\left( {{\rho_i},{{\pmb{\rho}}_{ -i}}} \right)\right\} }\right\rangle.
\end{equation}
It is worth mentioning that although the structure of the games formulated for the AF and DF networks is similar, their solution analyses are actually quite different. So we discuss them separately in the following subsections.

%\subsection{Solution of the Formulated Game}
\subsection{Existence of the Nash Equilibrium}
The most well-known solution to the non-cooperative games is the (pure strategy) Nash equilibrium (NE) \cite{Fudenberg_Game_book}. A NE of a given non-cooperative game $\left\langle  {{\cal N},\{ {{\bf{{\cal Q}}}_n}\} ,\left\{ {{{\cal U}_n}\left( {{{\bf{x}}_n},{{\bf{x}}_{ - n}}} \right)} \right\}} \right\rangle $ is a feasible point ${\bf{x}}^*$ such that
\begin{equation}
{\mathcal U}\left( {{\bf{x}}_n^*,{\bf{x}}_{ - n}^*} \right) \ge {\mathcal U}\left( {{{\bf{x}}_n},{\bf{x}}_{ - n}^*} \right),\;\forall {{\bf{x}}_n} \in {{\mathcal Q}_n}.
\end{equation}
In other words, a NE is a feasible strategy profile with the property that no single player can increase the utility by deviating from the strategy corresponding to the equilibrium, given the strategies of the other players. The following theorem proposed in \cite{Glicksberg_AMS_1952} is usually adopted to verify the existence of the NE:
\begin{theorem}\label{thm:existence_NE}
A NE \emph{exists} in the game $\left\langle  {{\cal N},\{ {{\bf{{\cal Q}}}_n}\} ,\left\{ {{{\cal U}_n}\left( {{{\bf{x}}_n},{{\bf{x}}_{ - n}}} \right)} \right\}} \right\rangle $ if $\forall n \in {\cal N}$, ${{\bf{{\cal Q}}}_n}$ is a \emph{compact and convex} set; ${{{\cal U}_n}\left(  {\bf{q}}\right)}$ is continuous in ${\bf{q}}$ and \emph{quasi-concave} in ${{\bf{q}}_n}$, where ${\bf{q}} = \left({{{\bf{q}}_n},{{\bf{q}}_{ - n}}}\right)$.
\end{theorem}

After investigating the properties of the action sets and the utility functions for the formulated games ${\mathcal{G}}_{AF}$ and ${\mathcal{G}}_{
DF}$, we have the following proposition regrading the existence of the NE:
\begin{proposition}\label{prop:NE_existence_AF_game}
The utility function ${u_i^{AF}}\left( {{\rho _i},{\pmb \rho _{ - i}}} \right)$ is quasi-concave in $\rho_i$ for any $i \in \mathcal{N}$. Moreover, the formulated power splitting game ${\mathcal{G}}_{AF}$ for the AF network possesses at least one NE. Moreover, the formulated power splitting game ${\mathcal G}_{DF}$ for the DF network also admits at least one NE.
\end{proposition}
\begin{proof}
See Appendix \ref{append:prop1}.
\end{proof}

\subsection{Uniqueness for the NE of the game ${\mathcal{G}}_{AF}$}
Once the NE is shown to exist, a natural question that arises is whether it is unique. This is important not only for predicting the state of the network but also crucial for convergence issues.
In principle, the uniqueness of the NE can be analyzed by several methods, which has been summarized in \cite{Lasaulce_SPM_2009}. However, since the formulated game ${\mathcal{G}}_{AF}$ is not a convex one, most of the methodologies cannot be applied except the standard function approach \cite{Yates_JSAC_1995} because, as shown below, it only requires that the best response function satisfies certain properties. To proceed, we first figure out the best response functions of the links (players), for which we have the following lemma:
\begin{lemma}\label{lemma:BR_AF}
Given a power splitting strategy profile $\pmb \rho$, the best response function of the link ${S_i} \to {R_i} \to {D_i}$ in the game ${\mathcal{G}}_{AF}$ can be expressed as (\ref{eq:BR_function_AF}) on top of next page.
\begin{figure*}
\begin{equation}\label{eq:BR_function_AF}
\begin{split}
{\mathcal B}_i^{AF}\left(\pmb \rho  \right) = \left\{ \begin{array}{l}
\frac{1}{2},\;\;\;\;\;\;\;\;\;\;\;\;\;\;\;\;\;\;\;\;\;\;\;\;\;\;\;\;\;\;\;\;\;\;\;\;\;\;\;\;\;\;\;\;\;{\rm{if}}\;\left( {{X_i} + {Y_i}} \right)\left( {{W_i} + 1} \right) = {Z_i} \\
\frac{{\sqrt {{\left( {{X_i} + {Y_i} + 1} \right)\left( {{W_i} + 1} \right)}} }}{{\sqrt {\left( {{X_i} + {Y_i} + 1} \right)\left( {{W_i} + 1} \right)}  + \sqrt {Z_i +W_i+1} }},\;\;{\rm{if}}\;\left( {{X_i} + {Y_i}} \right)\left( {{W_i} + 1} \right) \ne {Z_i} \\
  \end{array} \right..
\end{split}
\end{equation}
\hrulefill \vspace*{4pt}
\end{figure*}
\end{lemma}
\begin{proof}
See Appendix \ref{append:lemma_BR_AF}.
\end{proof}

Now let us verify the correctness of the best response function (\ref{eq:BR_function_AF}) by utilizing the special case that $\left( {{X_i} + {Y_i}} \right)\left( {{W_i} + 1} \right) = {Z_i}$. If we insert this condition into the expression of SINR (\ref{eq:SINR_AF}), we can readily obtain that the SINR is maximized when the term $\rho_i\left(1-\rho_i\right)$ is maximized. This implies that the best response $\rho_i^* = 1/2$, which is consistent with our previous analysis of the best response function. To gain more insights, let us rewrite $\left( {{X_i} + {Y_i}} \right)\left( {{W_i} + 1} \right) = {Z_i}$ as $ {{X_i} + {Y_i}}  = {Z_i}/\left( {{W_i} + 1} \right)$. Then we can note that the left-hand side of the previous equation represents the signal-plus-interference-to-noise ratio of the first hop when the relay only perform information forwarding (i.e., $\rho_i = 0$), while the right-hand side denotes the SINR of the second hop with only energy harvesting at the relay (i.e., $\rho_i = 1$). This special case reveals that the relay $R_i$ should equally split its received signal, when the signal-plus-interference-to-noise ratio of the first hop for $\rho_i = 0$ equals to the SINR of the second hop for $\rho_i = 1$.

We now define the vector function ${\pmb{\mathcal B}}^{AF}\left( \pmb \rho  \right) = \left[ {{{\mathcal B}_1^{AF}}\left( \pmb \rho  \right), \ldots ,{{\mathcal B}^{AF}_N}\left( \pmb \rho  \right)} \right]^T$. Then, according to the well-known fixed point theorem \cite{Lasaulce_SPM_2009}, the strategy profile ${\pmb \rho} ^*$ is a NE of the formulated game ${\mathcal{G}}_{AF}$ if and only if it is the fixed point of the function ${\pmb{\mathcal B}}^{AF}\left( \pmb \rho  \right)$ (i.e, ${\pmb{\mathcal B}}^{AF}\left( {\pmb \rho} ^*  \right) = {{\pmb \rho} ^* } $). Hence, the uniqueness for the NE of the formulated game is equivalent to that for the fixed point of the function ${\pmb{\mathcal B}}^{AF}\left( \pmb \rho  \right)$. Furthermore, it is shown in \cite{Yates_JSAC_1995} that the fixed point of the function ${\pmb{\mathcal B}}^{AF}\left( \pmb \rho  \right)$ is unique if ${\pmb{\mathcal B}}^{AF}\left( \pmb \rho  \right)$ is a \emph{standard} function. The standard function is defined as follows:
\begin{definition}\label{eq:standard_func}
A function ${\bf{ f}}\left( \bf x  \right)$ is said to be standard if it satisfies the following properties for all $\bf x \ge0$:
\begin{itemize}
  \item \emph{Positivity}: ${\bf{f}}\left( \bf x \right) > \bf{0}$.
  \item \emph{Monotonicity}: If ${\bf x} \ge {{\bf x} ^\prime}$, then ${\bf{f}}\left( \bf x  \right) \ge {\bf{f}}\left( {\bf x}^\prime  \right)$.
  \item \emph{Scalability}: For all $\alpha >1 $, $\alpha {\bf{f}}\left( \bf x  \right)> {\bf{f}}\left( \alpha {\bf x}  \right)$.
\end{itemize}
Here, all the inequalities are componentwise.
\end{definition}

After investigating the properties of the best response functions given in Lemma \ref{lemma:BR_AF}, we have the following proposition regarding the uniqueness of the NE:
\begin{proposition}\label{prop:NE_unique_AF}
The formulated game $\mathcal G_{AF}$ for the AF network always admits a unique NE.
\end{proposition}
\begin{proof}
See Appendix \ref{append:Prop_NE_unique_AF}.
\end{proof}

\subsection{Uniqueness for the NE of the game ${\mathcal{G}}_{DF}$}

To validate the uniqueness for the NE of the game ${\mathcal{G}}_{DF}$, most of the methodologies also fail to apply except the approach of standard function \cite{Yates_JSAC_1995} due to the non-differentiability of the utility functions (the $\min$ operator). Similarly, we first derive the best response functions of the links (players) in the game ${\mathcal G}_{DF}$ and have the following lemma:
\begin{lemma}\label{lemma:BR}
Given a power splitting strategy profile $\pmb \rho$, the best response function of the link ${S_i} \to {R_i} \to {D_i}$ in the game ${\mathcal G}_{DF}$ can be expressed as
\begin{equation}\label{eq:BR_function_DF}
\begin{split}
{{\mathcal B}_i^{DF}}\left(\pmb \rho  \right)&= \left[ {\left( {{X_i}{W_i} + {X_i} + {Y_i}{Z_i} + {Z_i}} \right) - } \right. \\
&\left. {\sqrt {{{\left( {{X_i}{W_i} + {X_i} - {Y_i}{Z_i} + {Z_i}} \right)}^2} + 4{Y_i}Z_i^2} } \right]/\left( {2{Y_i}{Z_i}} \right).
\end{split}
\end{equation}
\end{lemma}
\begin{proof}
See Appendix \ref{append:lemma_BR_function_DF}.
\end{proof}

We now define the vector function ${\pmb{\mathcal B}}^{DF}\left( \pmb \rho  \right) = \left[ {{{\mathcal B}_1^{DF}}\left( \pmb \rho  \right), \ldots ,{{\mathcal B}_N^{DF}}\left( \pmb \rho  \right)} \right]^T$. After investigating the properties of the function ${\pmb{\mathcal B}}^{DF}\left( \pmb \rho  \right)$, we have the following proposition regarding the uniqueness of the NE:
\begin{proposition}\label{prop:NE_unique}
The game ${\mathcal G}_{DF}$ also always possesses a unique NE.
\end{proposition}
\begin{proof}
See Appendix \ref{append:Prop_2}.
\end{proof}

\subsection{Distributed Algorithm}\label{sec:algorithm_AF}
So far, we have proved that the formulated games ${\mathcal {G}}_{AF}$ and ${\mathcal {G}}_{DF}$ for pure networks always have a unique NE for any channel conditions and network topologies. However, this equilibrium is meaningful in practice only if one can develop an algorithm that is able to achieve such an equilibrium from non-equilibrium states. In this subsection, we propose a distributed and iterative algorithm with provable convergence to achieve the NE, in which the links update their power splitting ratios simultaneously. In addition, we discuss the practical implementation of the proposed algorithm.

\subsubsection{Algorithm Description}
Various kinds of distributed algorithms have been proposed to find the NEs (see \cite{Lasaulce_book_2011} for more information). Here, we are interested in best response-based algorithms since we have obtained the best response functions of the formulated games. Relying on the derived best response functions in  (\ref{eq:BR_function_AF}) and (\ref{eq:BR_function_DF}), we develop a distributed power splitting algorithm for the pure networks, which is formally described in {\bf{Algorithm 1}}. In terms of the convergence of Algorithm 1, we have the following result:
\begin{proposition}\label{prop:convergence_alg_AF}
From any initial point, Algorithm 1 always converges to the unique NE of the formulated games $\mathcal G_{AF}$ and $\mathcal G_{DF}$.
\end{proposition}
\begin{proof}
Since the best response vector functions ${\pmb{\mathcal B}}^{AF}\left( \pmb \rho  \right)$ and ${\pmb{\mathcal B}}^{DF}\left( \pmb \rho  \right)$ are both standard (proved in Appendix \ref{append:Prop_NE_unique_AF} and \ref{append:Prop_2}), the proof of this proposition follows with reference to \cite[Thm. 2]{Yates_JSAC_1995}.
\end{proof}

\begin{table}
\begin{center}
\begin{algorithm}[H]
%\algsetup{linenosize=\normalsize}
%\normalsize
%\textbf{Input:}$\quad$ Positive integer $\upsilon$.\\
%\textbf{Output:}$\quad$ Binary matrices $\boldsymbol G_{2,q}$, $q=1,2,\cdots,2^{\upsilon-1}$.\\
%\textbf{Steps:}
\begin{algorithmic}[1]
\small
\STATE Set ${t} = 0$ and each player (link) $i\in{\mathcal N}$ chooses a random power splitting ratio ${\rho}_i{\left(0\right)}$ from the feasible set ${\mathcal A}_i$.
%\STATE {\textbf {for}} ${t} = 0:{\rm{Number\_of\_Iteration}}-1$
\STATE If a suitable termination criterion is satisfied: $\rm{STOP}$.
\STATE Each link $i \in {\mathcal N}$ updates the power splitting ratio via executing
\begin{equation}\label{eq:BR_algorith}
%\footnotesize
\begin{split}
{\rho _i}(t + 1) = \left\{ \begin{array}{l}
 {\cal B}_i^{AF}\left( {\pmb \rho (t)} \right)\;{\rm{for\;AF\;network}} \\
 {\cal B}_i^{DF}\left( {\pmb \rho (t)} \right)\;{\rm{for\;DF\;network}} \\
 \end{array} \right..
\end{split}
\end{equation}
%\STATE
%\begin{adjustwidth}{0.5cm}{0cm}
%The energy provider determines the value of $\gamma(t+1)$, and broadcasts that value as well as ${\bf{q}}_{\Sigma}{\left( {{t+1}} \right)} = \sum\nolimits_{n = 1}^N {{\bf{q}}_n{\left( {{t+1}} \right)}} $ to the consumers.
%\end{adjustwidth}
\STATE
%\begin{adjustwidth}{0.5cm}{0cm}
$t \leftarrow t + 1$; go to $\rm{STEP}$ 2.
%\end{adjustwidth}
%\STATE {\textbf {Until}} the termination criterion is satisfied.
\end{algorithmic}
\caption{: Distributed Power Splitting Algorithm for the Pure Networks
}
\end{algorithm}
\end{center}
%\caption{Description of \textbf{Algorithm $1$} that is used to generate a series of full column rank binary matrices.}\label{fig:alg1}
\end{table}

\subsubsection{Implementation Discussion}
Note that the distributed nature of the above algorithm is based on modeling each link as a single player. However, each link consists of three nodes (i.e., source, relay, and destination), which are geographically separated in practical networks. Thus, an efficient implementation of the proposed algorithm with the minimum information sharing should be designed.

In our design, the relay node is expected to be the link coordinator that undertakes the information collection and the best-response computation. Here, we assume that the energy consumed for the algorithm computations at the relay nodes are negligible compared with the harvested energy used for information forwarding since these computations are quite simple. This assumption can be further supported by the rapid development of the low-power chips. Next, we aim to identify the information that is needed to collect or exchange for the implementation of the proposed algorithm. According to the best response functions given in (\ref{eq:BR_function_AF}), the relay $R_i$ needs to know the values of the parameters $X_i$, $Y_i$, $Z_i$, and $W_i$ defined in (\ref{eq:def_XYZW}). To this end, the relay $R_i$ should perform the following tasks:
\begin{itemize}
  \item {\emph{Task 1}}: Measure the channel gains $g_{ii}$ and $h_{ii}$, acquire the transmission power $P_i$ from its source, and then calculate the value of $X_i$ based on these information.
  \item {\emph{Task 2}}: Measure\footnote{The measurement of the signal power can be performed by the radio scene analyzer \cite{Shi_TWC_2009}.} the power of its received signal, acquire the power of the received signal at the destination $D_i$, and then calculate the values of $Y_i$, $Z_i$, $W_i$.
  \item {\emph{Task 3}}: Calculate the optimal power splitting ratio $\rho_i$ based on (\ref{eq:BR_function_AF}).
\end{itemize}

From the above description, we can observe that the overheads are required in transmitting the following three kinds of information for each link in the proposed algorithm: (1) pilots for estimating the channel state information (CSI) from source to its dedicated relay and CSI from relay to destination, (2) the value of transmit power from the source to relay, and (3) the power of received signal at the destination, which needs to be sent from the destination to its dedicated relay. The first two kinds of overheads are only needed once for each channel realization, while the third one is required in each iteration of the proposed algorithm. From the above discussion, we can see that in the proposed algorithm, only some local information needs to be exchanged within each link and no information needs to be exchanged among different links.

Finally, note that a possible termination criterion for the proposed Algorithm 1 can be
 \[\left[ {{\rho _i}\left( {t + 1} \right) - {\rho _i}\left( t \right)} \right]/{\rho _i}\left( {t + 1} \right) \le \zeta ,\]
 where $\zeta$ is a sufficiently small constant.

\section{Extension to Hybrid Network}
In this section, we generalize the proposed game-theoretical power splitting scheme to a hybrid network containing both AF and DF relays. The set of the links with AF relays and DF relays are denoted by ${\mathcal N}_{AF}$ and ${\mathcal N}_{DF}$, respectively. Before formally describing the non-cooperative game for this case, it is important to notice that for a given power splitting ratio profile $\pmb\rho_{-i}$, the amount of interference received by the destination $D_i$ is fixed. This is because for a given realization of channels, the interference at each destination is only determined by the transmit powers of relays in other links and is independent of the relaying protocols (i.e., AF or DF) adopted by these relays. Thus, the achievable rate of the $i$th link in the hybrid network can still be expressed by (\ref{eq: utiliy_AF}), when the AF relaying protocol is adopted at the relay $R_i$, and by (\ref{eq:utility_DF}), when the DF relaying protocol is adopted at the relay $R_i$.

Now, we can formulate the following non-cooperative game for the considered hybrid network:
\begin{itemize}
\item \emph{{Players}}: The $N$ $S$-$R$-$D$ links.
\item \emph{{Actions}}: Each link determines its power splitting ratio $\rho_i \in {\mathcal{A}}_i$ to maximize the achievable rate for its own link.
\item \emph{{Utilities}}: The achievable rate
\end{itemize}
\begin{equation}\label{}
    u_i^{HD}\left( {{\rho _i},{\pmb \rho _{ - i}}} \right) = \left\{ \begin{array}{l}
    u_i^{AF}\left( {{\rho _i},{\pmb \rho _{ - i}}} \right)~{\rm{in}}~(\ref{eq: utiliy_AF}),\;\;\;\;{\rm{if}}\;i \in {{\mathcal N}_{AF}} \\
    u_i^{DF}\left( {{\rho _i},{\pmb \rho _{ - i}}} \right)~{\rm{in}}~(\ref{eq:utility_DF}),\;\;{\rm{if}}\;i \in {{\mathcal N}_{DF}} \\
    \end{array} \right..
\end{equation}

For convenience, we denote the above non-cooperative game as
\begin{equation}\label{eq:hybrid_game}
{\mathcal{G}}_{HD} = \left\langle {\mathcal{N},\left\{ \mathcal{A}_i  \right\},\left\{{{u}_i^{HD}}\left( {{\rho_i},{{{\pmb \rho}}_{ -i}}} \right)\right\} }\right\rangle.
\end{equation}

Subsequently, according to Proposition \ref{prop:NE_existence_AF_game}, Lemma \ref{lemma:BR_AF}, Proposition \ref{prop:NE_unique_AF}, Lemma \ref{lemma:BR}, and Proposition \ref{prop:NE_unique}, we have the following corollary regarding the best response function and the existence and uniqueness of the NE for the game ${\mathcal G}_{HD}$:
\begin{corollary}
The best response function of the $i$th link in the game ${\mathcal G}_{HD}$ can be expressed as
\begin{equation}\label{eq:BR_hybrid}
{{\mathcal B}_i^{HD}}\left(\pmb \rho  \right) = \left\{ \begin{array}{l}
{{\mathcal B}_i^{AF}}\left(\pmb \rho  \right)~{\rm{in}}~(\ref{eq:BR_function_AF}),\;\;{\rm{if}}\;i \in {{\mathcal N}_{AF}} \\
{{\mathcal B}_i^{DF}}\left(\pmb \rho  \right)~{\rm{in}}~(\ref{eq:BR_function_DF}),\;\;{\rm{if}}\;i \in {{\mathcal N}_{DF}} \\
\end{array} \right..
\end{equation}
Furthermore, the game ${\mathcal G}_{HD}$ always possesses one and only one NE.$~~~~~~~~~~~~~~~~~~~~~~~~~~~~~~~~~~~~~~~~~~~~~~~~~~~~~~~~~~~~~~~~~~~~~~~~~~~~~~~\square$
\end{corollary}

Then, replacing the best response update (\ref{eq:BR_algorith}) in Algorithm 1 by the one in (\ref{eq:BR_hybrid}), we can have a best response-based distributed algorithm for the links to achieve the NE of the game ${\mathcal G}_{HD}$. This algorithm is referred to {\bf Algorithm 2}, which is omitted here due to its similarity with Algorithm~1.

\section{Numerical Results}
\begin{figure*}
\centering
 \subfigure[Best responses of the game ${\mathcal G}_{AF}$]
  {\scalebox{0.35}{\includegraphics {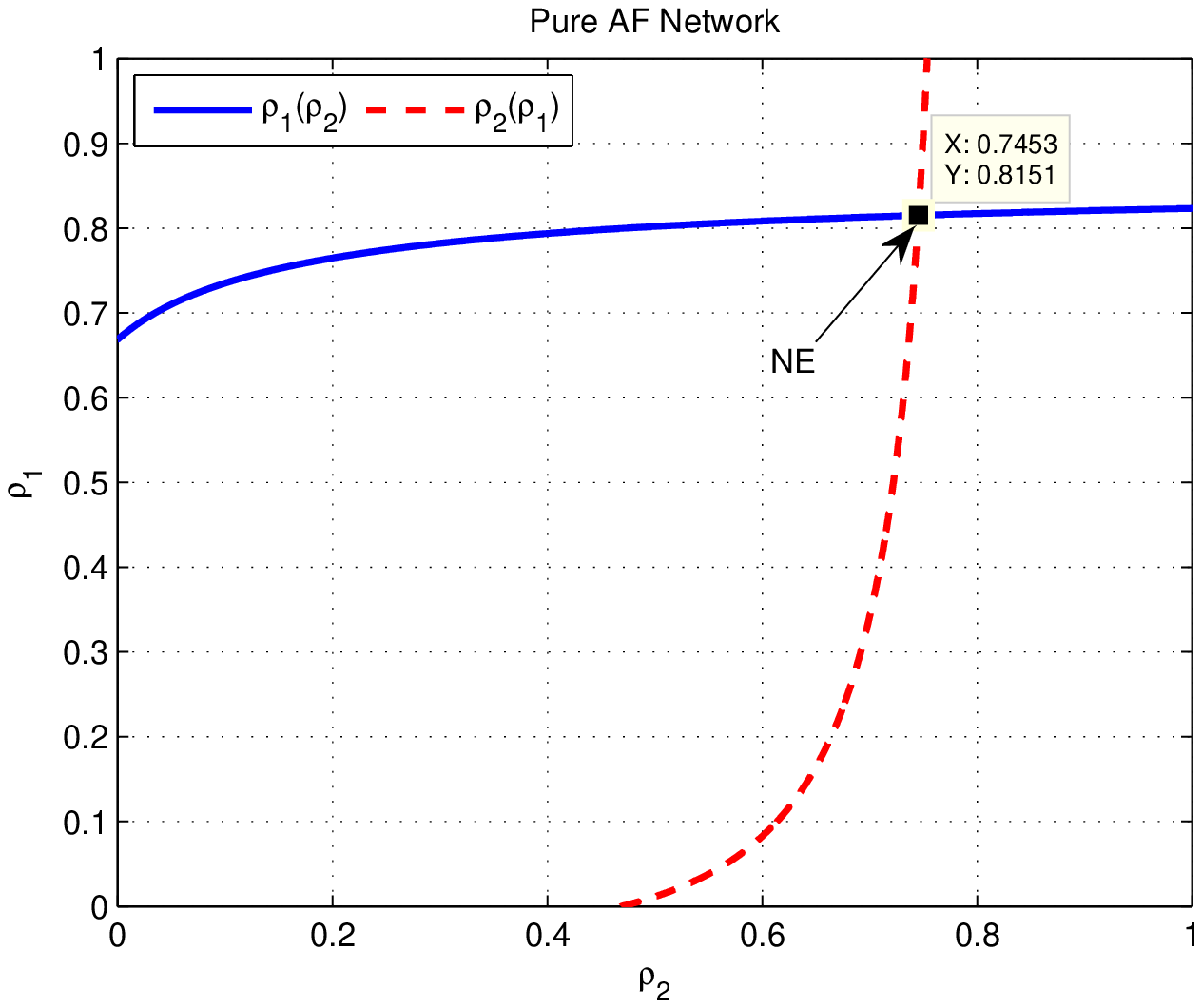}
  \label{fig:NE_2_link_a}}}
\hfil
 \subfigure[Best responses of the game ${\mathcal G}_{DF}$]
  {\scalebox{0.35}{\includegraphics {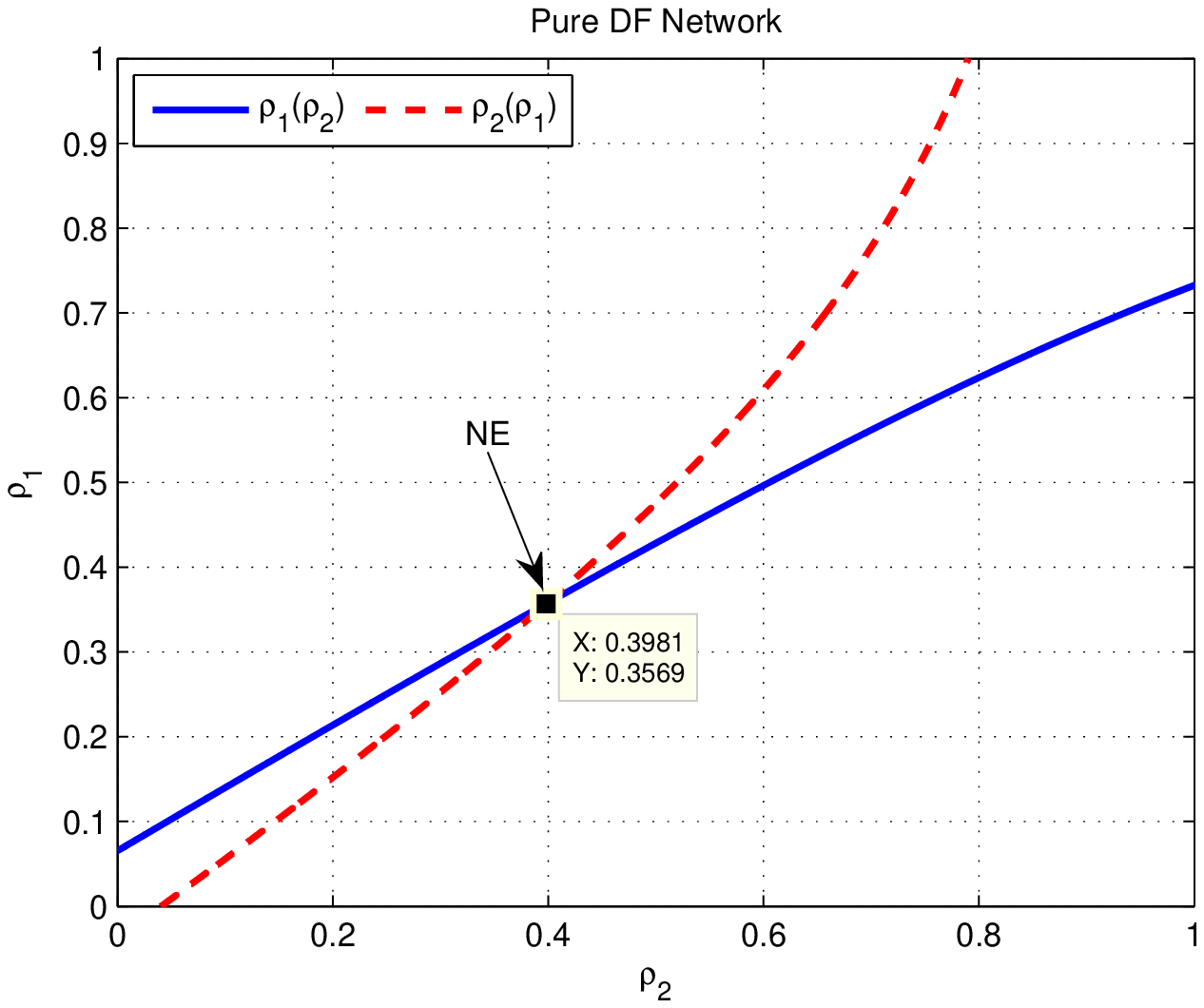}
\label{fig:NE_2_link_b}}}
\hfil
 \subfigure[Best responses of the game ${\mathcal G}_{HD}$]
  {\scalebox{0.35}{\includegraphics {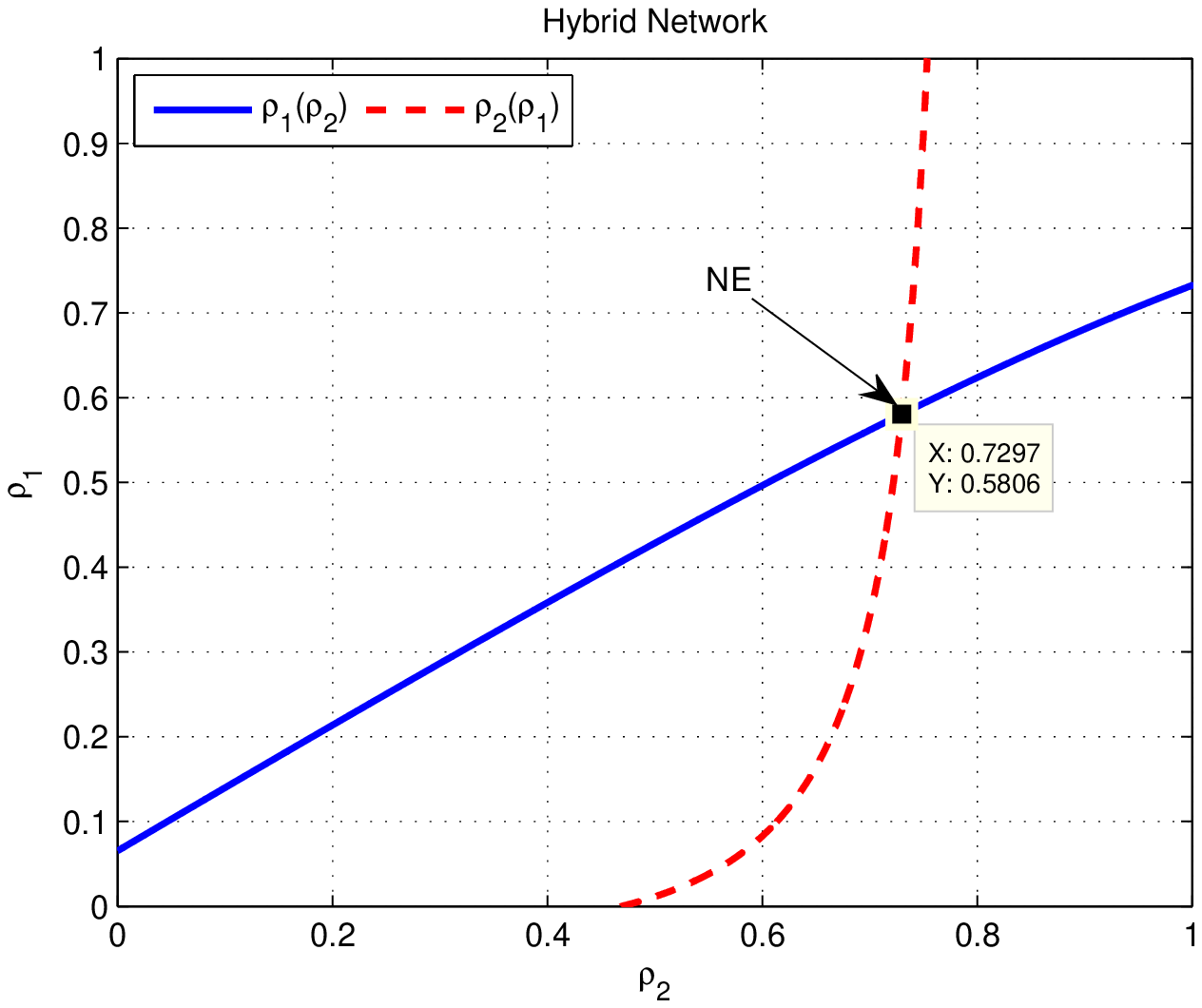}
\label{fig:NE_2_link_c}}}
\hfil
 \subfigure[Convergence of Algorithm 1]
  {\scalebox{0.35}{\includegraphics {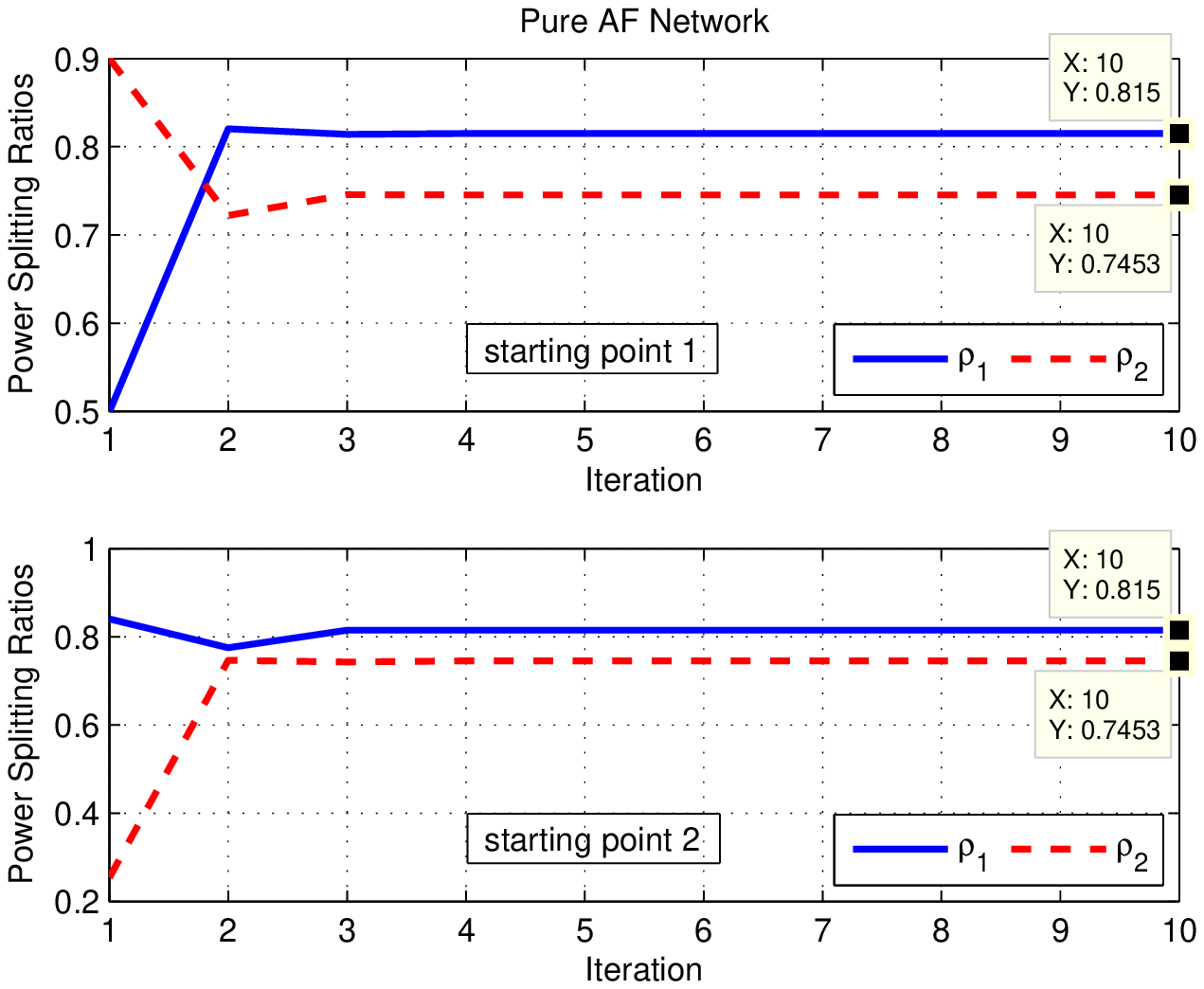}
  \label{fig:NE_2_link_d}}}
\hfil
 \subfigure[Convergence of Algorithm 1]
  {\scalebox{0.35}{\includegraphics {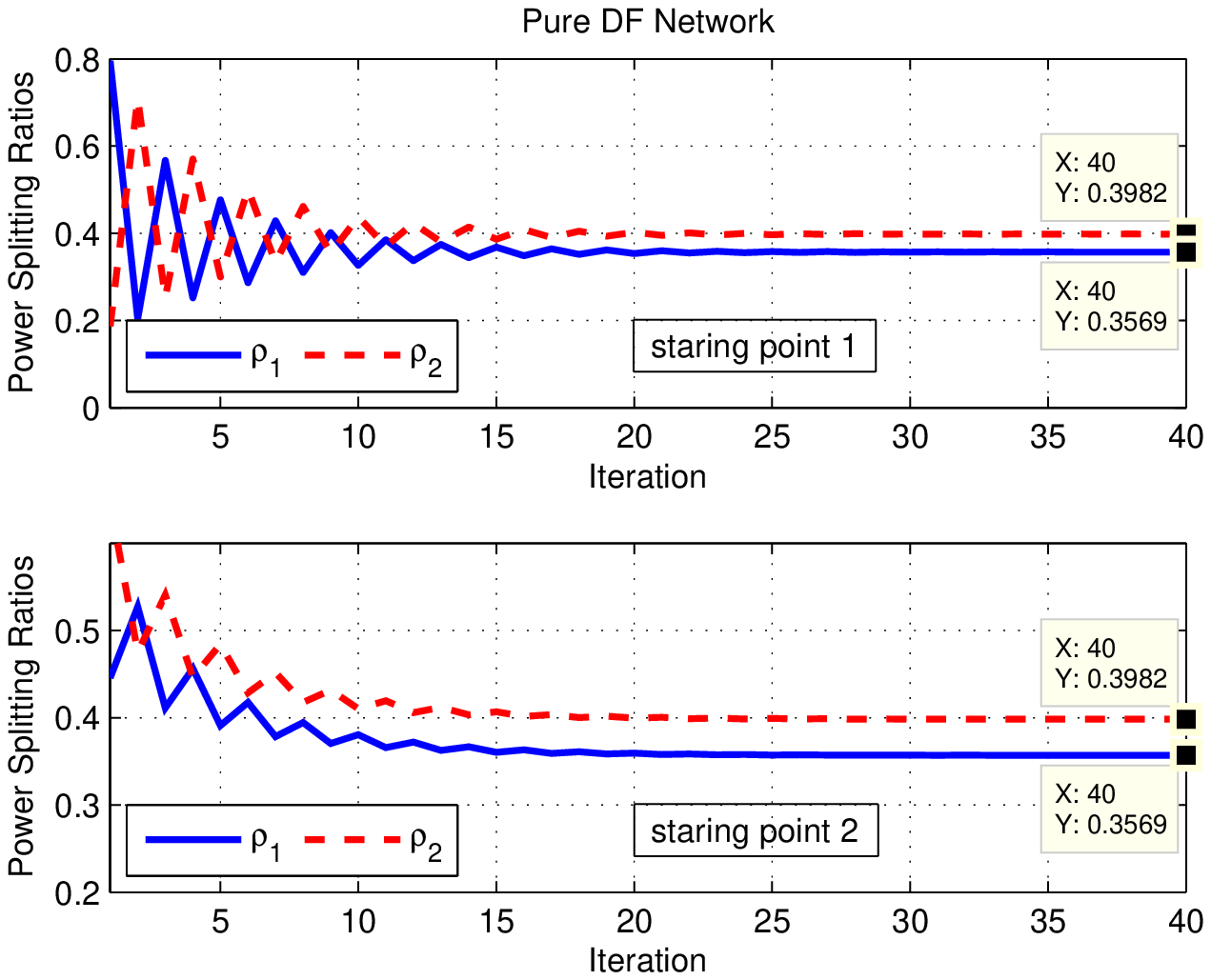}
\label{fig:NE_2_link_e}}}
\hfil
 \subfigure[Convergence of Algorithm 2]
  {\scalebox{0.35}{\includegraphics {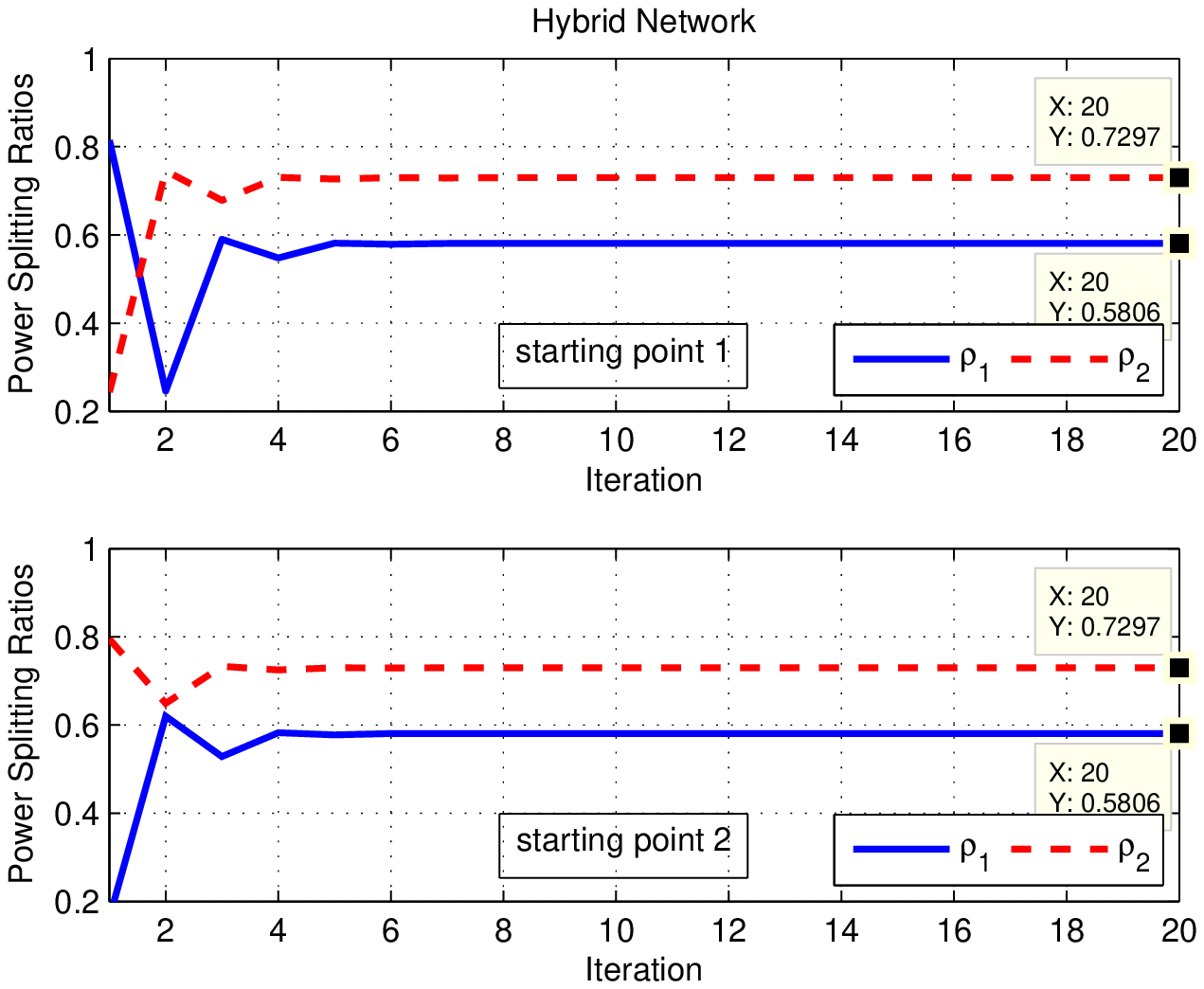}
\label{fig:NE_2_link_f}}}
\caption{The best response functions of the three formulated games and the convergence of Algorithm 1-2 from different starting points in the considered two-link network with parameters given in (\ref{eq:simulation_param_2_link_1})-(\ref{eq:simulation_param_2_link_3}).}
\label{fig:NE_2_link}
\end{figure*}

In this section, we will present some numerical results to illustrate and validate the above theoretical analyses. To obtain meaningful results, we restrict our attention to a linear topology for each link. Specifically, $S_i$-$R_i$-$D_i$ forms a straight line with unit length, i.e., $d_{S_i D_i}=  d_{S_i R_i} + d_{R_i D_i} = 1, \forall i\in {\mathcal N}$, with $d_{XY}$ denoting the distance between nodes $X$ and $Y$. Note that in order to obtain meaningful insights into the system performance, we treat the node distances as constants in this paper. In practical systems, the spatial node distributions \cite{Ding_SPL_2013_Coop,Ding_TWC_2014_Wireless} should be considered, but it will require a new analytical framework, which is out of the scope of this paper and we would like to consider it as our future work. The channels between all transceiver pairs are assumed to be subject to mutually independent Rayleigh fading. To take into account the impact of path loss, we adopt the channel model that ${\mathbb{E}}\left\{\left|g_{ij}\right|^2\right\} = \left({d_{S_i R_j}}\right)^{ - \tau }$ and ${\mathbb{E}}\left\{\left|h_{ij}\right|^2\right\} = \left({d_{R_i D_j}}\right)^{ - \tau }$, where $\tau \in [2,5]$ is the path loss factor \cite{Chen_SPL_2010}. In all the simulations, without loss of generality, we set $\sigma^2 = 1$ and $\eta = 0.5$.

We first consider a simple network consisting of two links with the following randomly generated parameters:
\begin{equation}\label{eq:simulation_param_2_link_1}
\left[P_1,P_2\right] = \left[5.3080,7.1917\right],
\end{equation}
\begin{equation}\label{eq:simulation_param_2_link_2}
\left[ {\begin{array}{*{20}{c}}
   {{{\left| {{g_{11}}} \right|}^2}} & {{{\left| {{g_{12}}} \right|}^2}}  \\
   {{{\left| {{g_{21}}} \right|}^2}} & {{{\left| {{g_{22}}} \right|}^2}}  \\
\end{array}} \right] = \left[ {\begin{array}{*{20}{c}}
   {2.1713} & {1.4836}  \\
   {3.0937} & {0.9773}  \\
\end{array}} \right],
\end{equation}
\begin{equation}\label{eq:simulation_param_2_link_3}
\left[ {\begin{array}{*{20}{c}}
   {{{\left| {{h_{11}}} \right|}^2}} & {{{\left| {{h_{12}}} \right|}^2}}  \\
   {{{\left| {{h_{21}}} \right|}^2}} & {{{\left| {{h_{22}}} \right|}^2}}  \\
\end{array}} \right] = \left[ {\begin{array}{*{20}{c}}
   {0.4475} & {1.5760}  \\
   {1.5406} & {2.6081}  \\
\end{array}} \right].
\end{equation}

Both links in the game ${\mathcal G}_{AF}$ (${\mathcal G}_{DF}$) adopt the AF (DF) relaying protocol, while link $1$ implements the DF relaying scheme and link $2$ implements the AF relaying scheme in the game ${\mathcal G}_{HD}$. In Fig. \ref{fig:NE_2_link_a}-\ref{fig:NE_2_link_c}, we first plot the best response functions (i.e., $\rho_1\left(\rho_2\right)$ and $\rho_2\left(\rho_1\right)$) of the three games, respectively. With reference to \cite{Lasaulce_SPM_2009}, the intersection points of the best response functions are actually the NEs of the corresponding games. From Fig. \ref{fig:NE_2_link_a}-\ref{fig:NE_2_link_c}, we can see that the two curves only admit one intersection point in all cases, which indicates that all the formulated games in the considered two-link network possess a unique NE. In Fig. \ref{fig:NE_2_link_d}-\ref{fig:NE_2_link_f}, the proposed algorithms are executed to achieve the corresponding NEs from randomly generated initial points. It can be seen from Fig. \ref{fig:NE_2_link_d}-\ref{fig:NE_2_link_f} that the proposed algorithms can converge to the corresponding NEs obtained in Fig. \ref{fig:NE_2_link_a}-\ref{fig:NE_2_link_c} from different starting points. Hence, the observations from Fig. \ref{fig:NE_2_link_a}-\ref{fig:NE_2_link_f} verify the correctness of our theoretical analyses.

To show that the proposed algorithm can also converge to the NE in multi-link scenarios, we illustrate its convergence performance for an example of a randomly generated four-link hybrid network in Fig. \ref{fig:hybrid_4_link}. As can be observed from Fig. \ref{fig:hybrid_4_link}, the proposed algorithm can converge to the same values (i.e., the NE) from two different initial points, which further validates the theoretical analyses. Note that due to the space limitation, we only show results in Figs. \ref{fig:NE_2_link}-\ref{fig:hybrid_4_link} for one random realization of transmit powers and channel gains, although similar results can also be shown for other realizations.

Next, we investigate the average sum-rate of the relay interference networks that implement the proposed power splitting scheme. We consider both a two-link scenario and a multi-link scenario. For simplification and illustrative purpose, we consider that all links in the considered networks are mutually parallel. Thus, the two-link setting is characterized by the interlink distance, denoted by $d_L$. In the multi-link setting, we assume that the distances between adjacent links are equal. Hence, this scenario is characterized by the number of links $N$ and the distance between the farthest two links, denoted by $d_{\max}$. In the following figures, each curve is plotted by averaging over the results from $10000$ independent channel realizations. In addition, it is straightforward to imagine that the performance curves of the hybrid network should be located between that of the AF and DF network. Thus, we will omit the curves of hybrid networks in the following to avoid crowded figures.

Fig. \ref{fig:average_sum_rate_2_link} illustrates the average sum-rate of a two-link network, where the performances are compared across the optimal sum-rate obtained by the centralized optimization problem (\ref{eq:centralized_optimization_AF}) solved via exhaustive search, the proposed game-theoretical approach, the random scheme in which the power splitting ratios are randomly generated over $[0,1]$. Both AF and DF networks with symmetric and asymmetric topologies are simulated. We can observe from Fig. \ref{fig:average_sum_rate_2_link} that the proposed game-theoretical method outperforms the the random scheme in the AF network when $d_L$ exceeds a certain value, and is always superior to the random scheme in the DF network for all cases. In addition, the performance gap between these two schemes becomes larger when the inter-link distance $d_L$ increases. Moreover, it can be seen from Fig. \ref{fig:average_sum_rate_2_link} that the proposed game-theoretical approach suffers performance loss compared to the optimal scheme when the inter-link distance is very small, i.e., in high interference scenarios. However, as the inter-link distance increases, it approaches the centralized optimal scheme in the AF network and quickly coincides with the centralized optimal scheme in the DF network. Therefore, we can claim that the proposed game-theoretical approach can achieve a near-optimal performance on average, especially for the scenarios with medium and large interlink distance (i.e., relatively low and moderate interference).
\begin{figure}
\centering \scalebox{0.5}{\includegraphics{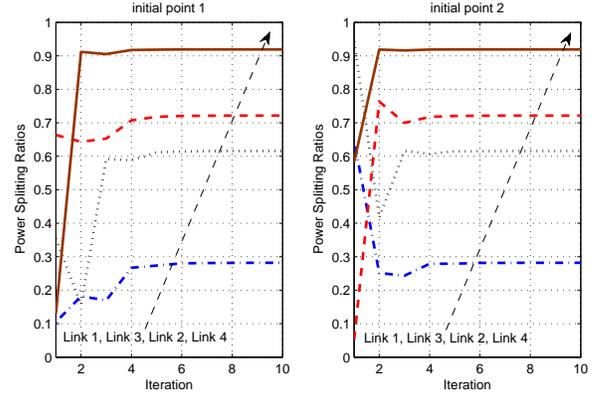}}
\caption{Convergence of the proposed algorithm for a randomly generated four-link hybrid network with two different initial points, ${\mathcal N}_{DF} = [1\;3]$ and ${\mathcal N}_{AF} = [2\; 4]$.\label{fig:hybrid_4_link}}
\end{figure}

\begin{figure*}
\centering
 \subfigure[AF network]
  {\scalebox{0.5}{\includegraphics {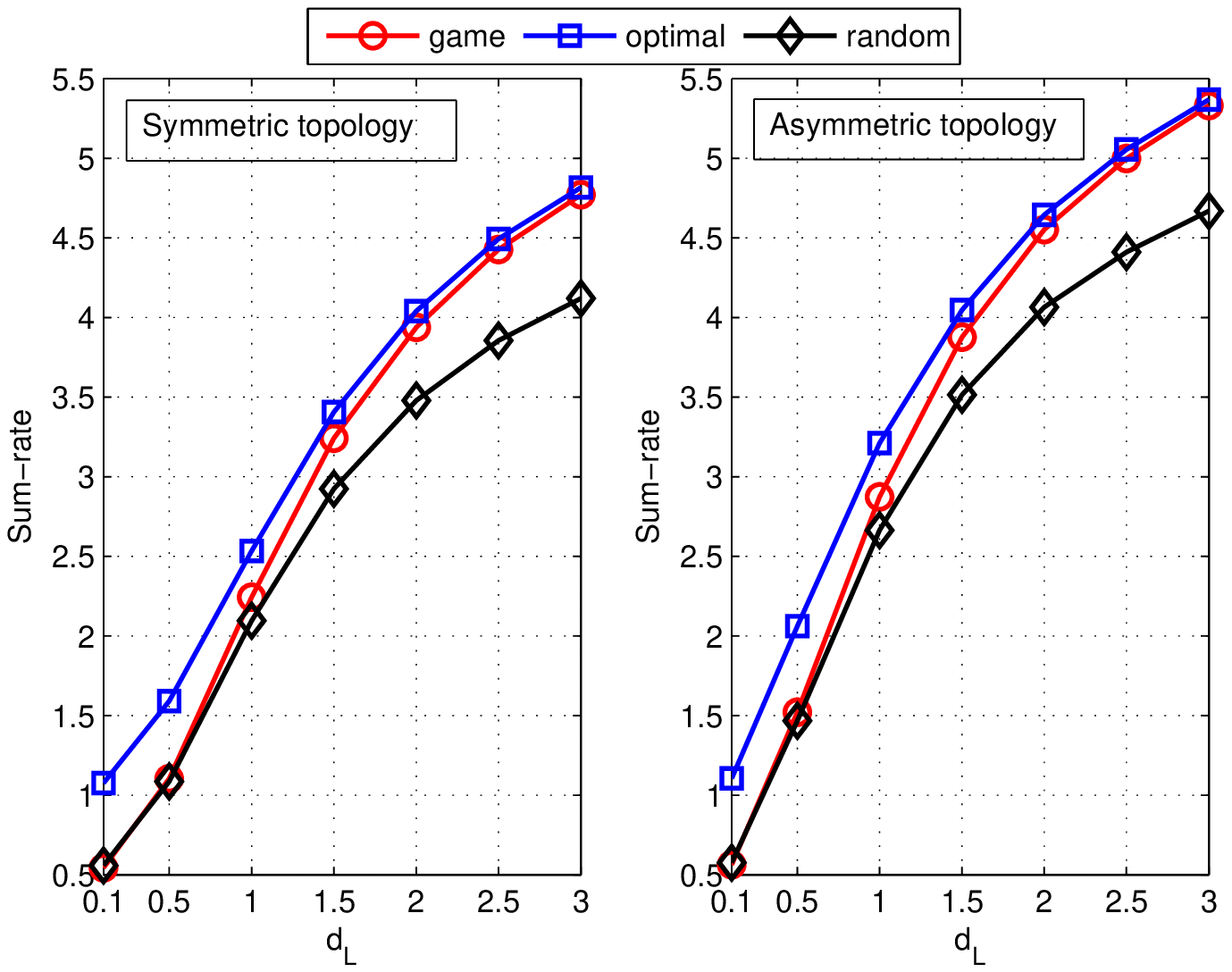}
  \label{fig:sum_rate_AF}}}
\hfil
 \subfigure[DF network]
  {\scalebox{0.5}{\includegraphics {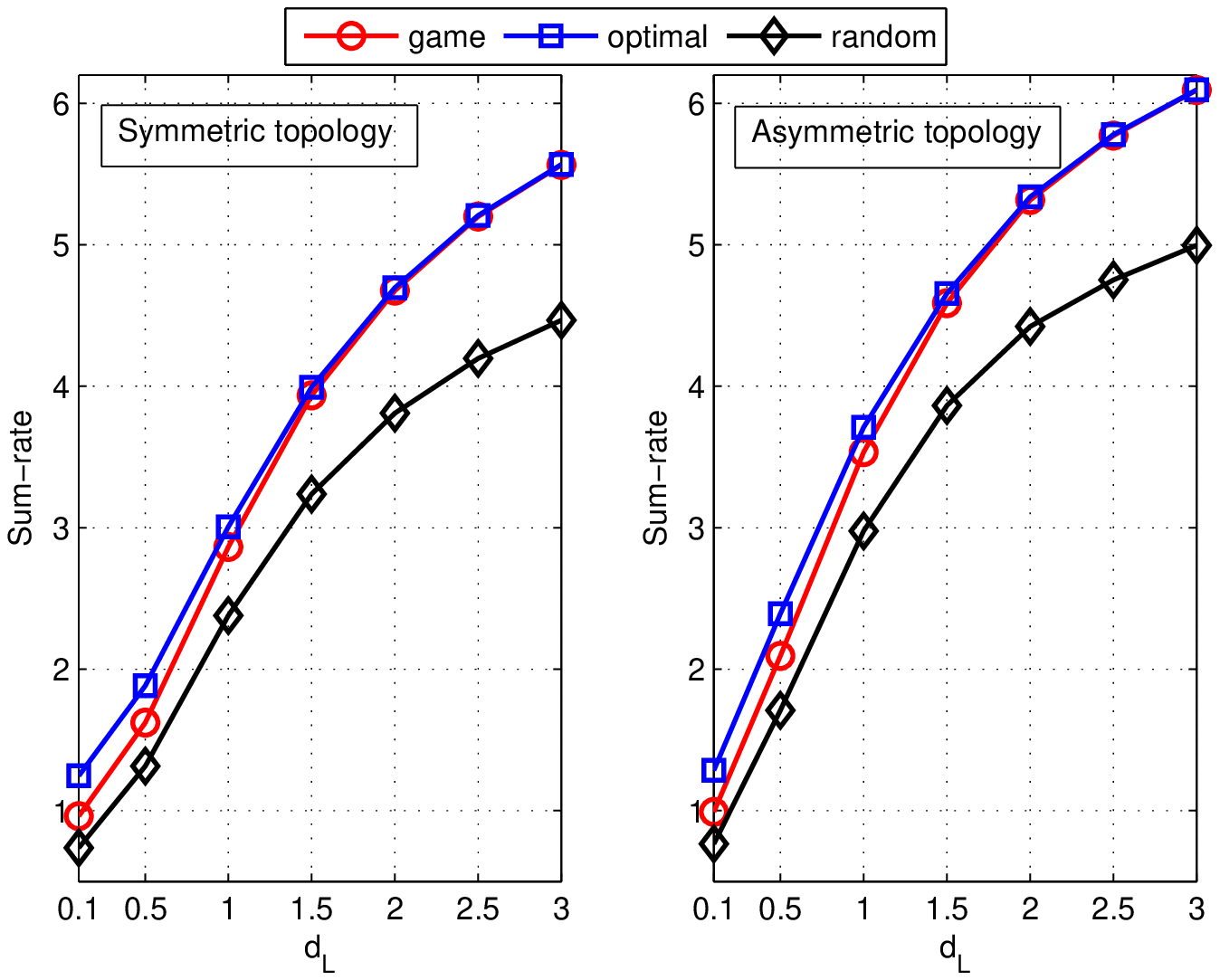}
\label{fig:sum_rate_DF}}}
\caption{The average sum-rates of two-link AF and DF networks for (a) symmetric network with $d_{S_i R_i} = d_{R_i D_i} = 0.5$ and $P_i = 15~ {\rm{dB}},~\forall i = 1,2$, (b) asymmetric network with $d_{S_1 R_1} = d_{R_2 D_2} = 0.25$ and $P_i = 15~ {\rm{dB}},~\forall i = 1,2$.}
\label{fig:average_sum_rate_2_link}
\end{figure*}

\begin{figure*}
\centering
 \subfigure[Average sum-rate]
  {\scalebox{0.5}{\includegraphics {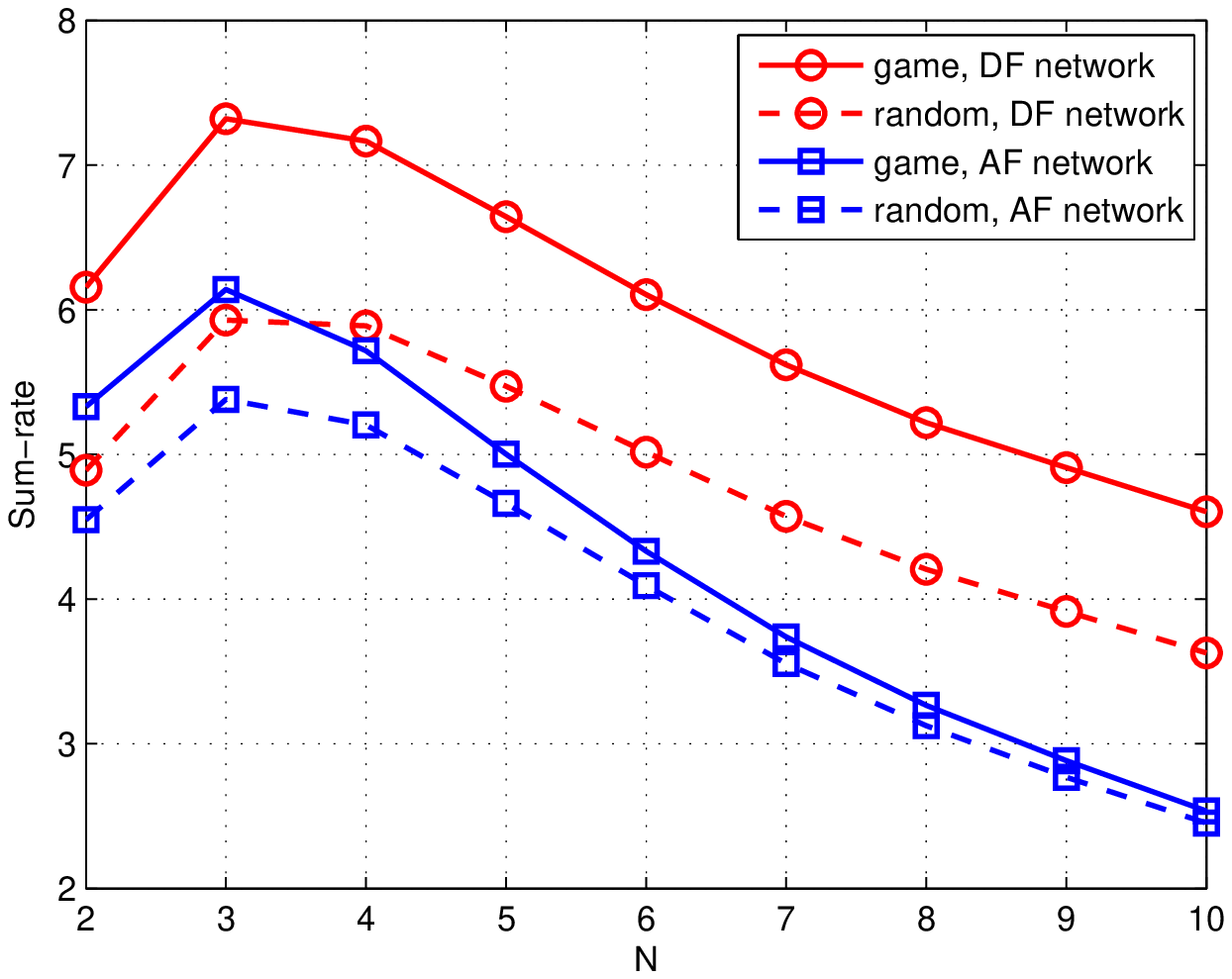}
  \label{fig:impact_link_num_sum_rate}}}
\hfil
 \subfigure[Average power splitting ratio]
  {\scalebox{0.5}{\includegraphics {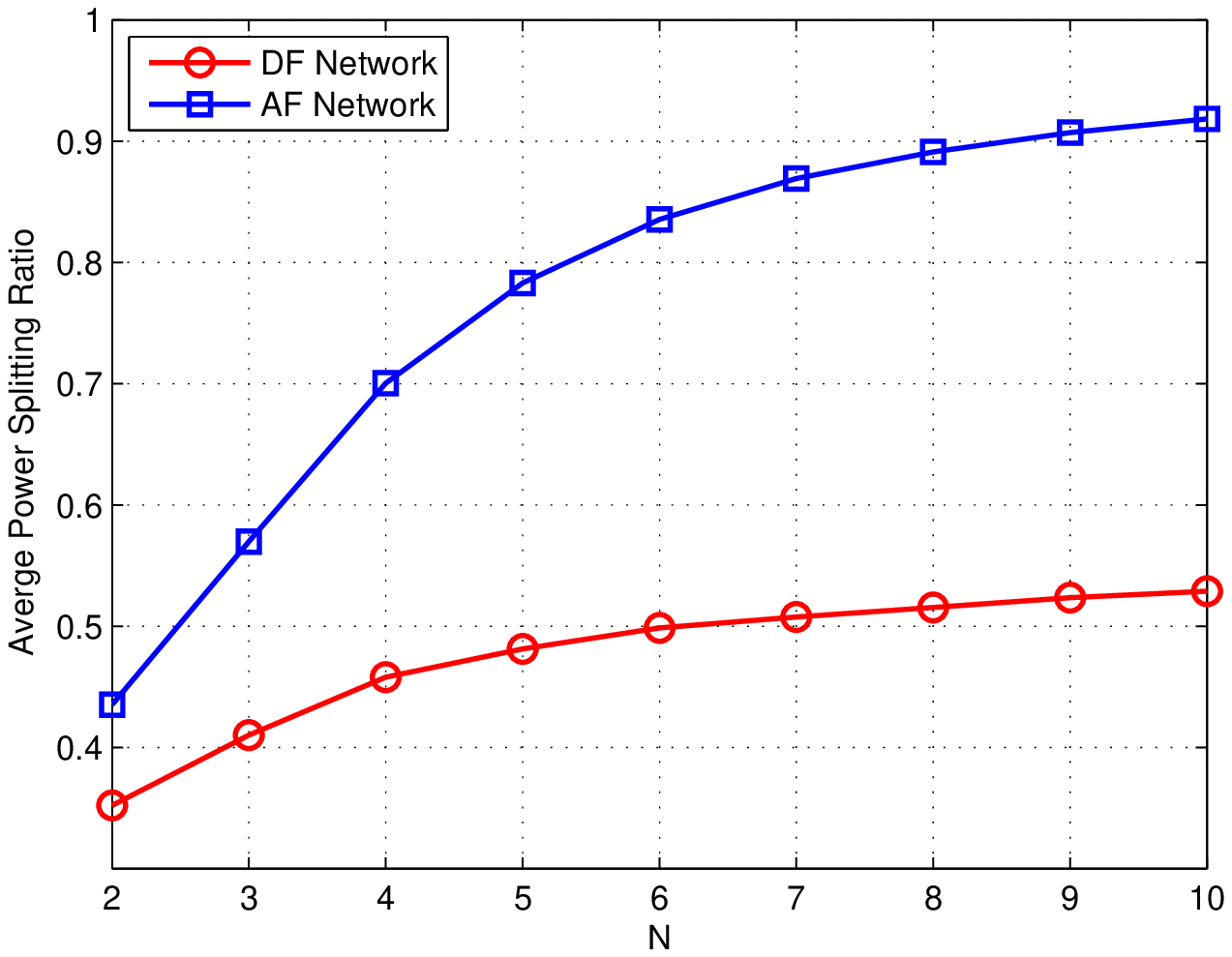}
\label{fig:impact_link_num_rho}}}
\caption{The impact of link numbers on (a) average sum-rate and (b) average power splitting ratio in AF and DF networks with $d_{\max} = 5$ and $P_i = 15$ (dB), $\forall i$.}
\label{fig:impact_link_num}
\end{figure*}

\begin{figure*}
\centering
 \subfigure[Average sum-rate]
  {\scalebox{0.5}{\includegraphics {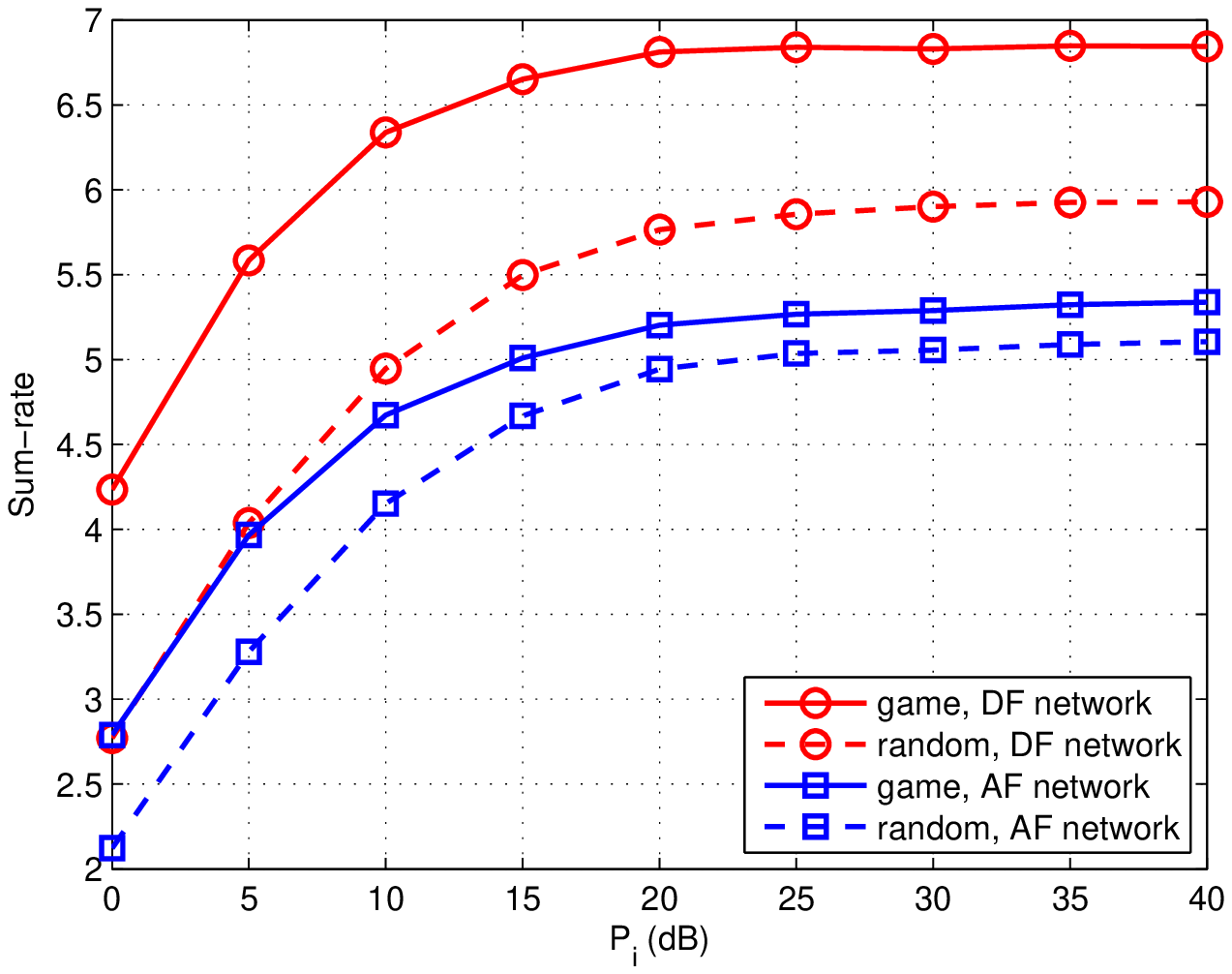}
  \label{fig:impact_SNR_sum_rate}}}
\hfil
 \subfigure[Average power splitting ratio]
  {\scalebox{0.5}{\includegraphics {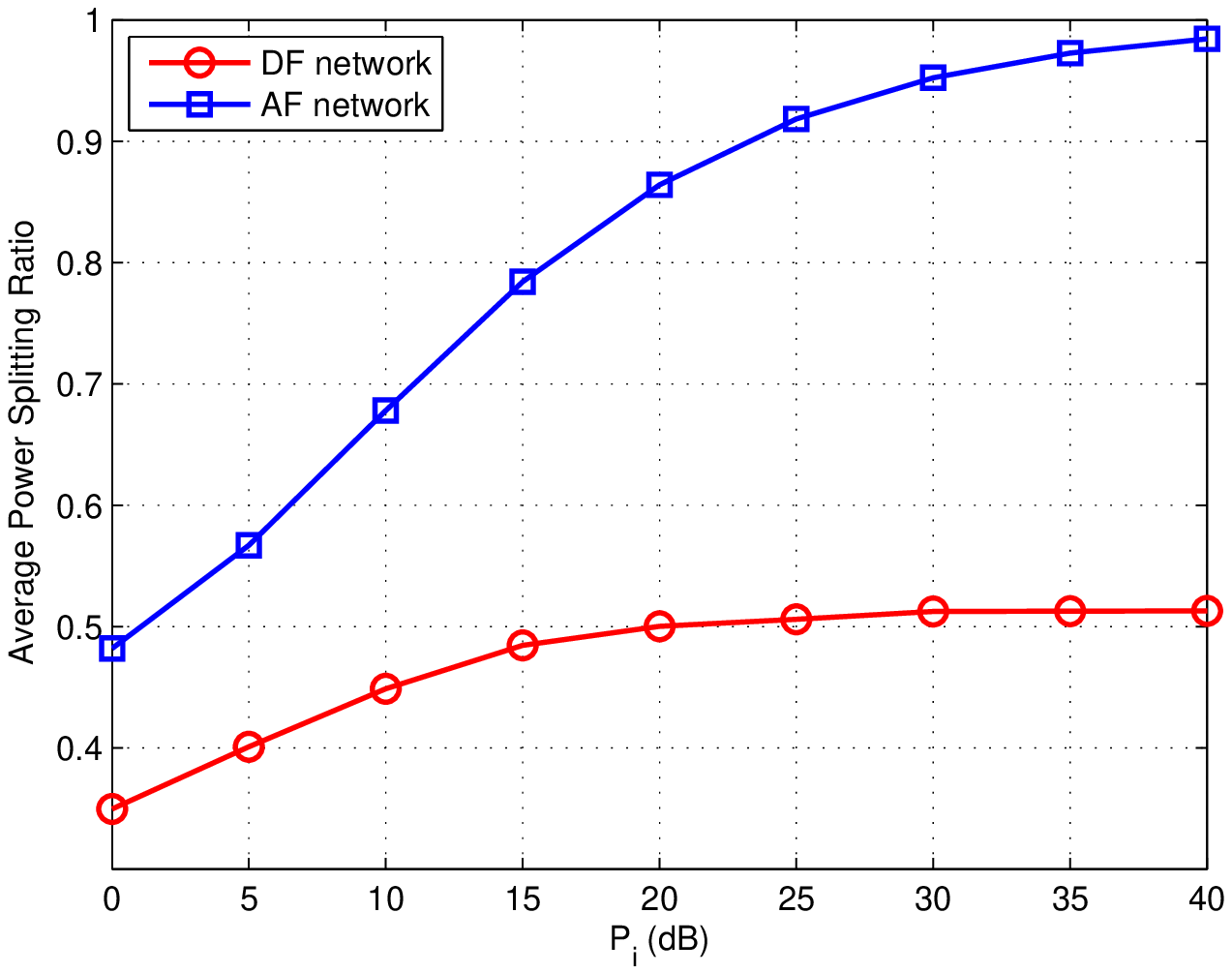}
\label{fig:impact_SNR_rho}}}
\caption{The impact of the transmit powers on (a) average sum-rate and (b) average power splitting ratio in AF and DF networks with $d_{\max} = 5$ and $N = 5$.}
\label{fig:impact_SNR}
\end{figure*}

Figs. \ref{fig:impact_link_num} and \ref{fig:impact_SNR} respectively demonstrate the impacts of the number of links and the transmit powers on the average sum-rates and average power splitting ratios of AF and DF networks in the multi-link scenario, where all relays are assumed in middle positions (i.e., $d_{S_i R_i} = d_{R_i D_i} = 0.5,~\forall i$). In Fig. \ref{fig:impact_link_num} (a), we plot the average sum-rate curves of the AF and DF networks achieved by the proposed game-theoretical scheme and the random scheme versus the number of links with $d_{\max} = 5$ and $P_i = 15$ dB, $\forall i$. Note that the performance of the centralized optimal scheme is omitted here because the corresponding optimization problem is non-convex and thus cannot be efficiently solved in a multi-link scenario. From Fig. \ref{fig:impact_link_num} (a), we observe that the average sum-rates of both AF and DF networks first increase and then keep decreasing with the growth of the number of links. This observation is understandable. Specifically, the initial sum-rate increase is actually a multiplexing gain as more links share the same spectrum with relatively low mutual interference. Nevertheless, with a further increase in the number of links, the interlink interference becomes stronger, which leads to a monotonically decreasing sum-rate. We can also see from Fig. \ref{fig:impact_link_num} (a) that the proposed game-theoretical approach always outweighs the random method in both AF and DF networks for all cases. The performance gap between these two schemes in the DF network is significantly larger than that in the AF networks. This indicates that the proposed game-theoretical approach can achieve a higher performance improvement in a DF network than the one in a AF network. In addition, it can be observed from Fig. \ref{fig:impact_link_num} (a) that the performance gap decreases gradually as the number of links increases, which is due to a more severe interference. Fig. \ref{fig:impact_link_num} (b) depicts the average power splitting ratios\footnote{Since plotting the curve for each link's power splitting ratio will create too many curves, which is difficult to illustrate, we choose the average power splitting ratio of all links as the performance metric for the purpose of facilitating the illustration. This metric can effectively reflect the changes of the ratios in the majority of links.} of both AF and DF networks, when the proposed game-theoretical scheme is implemented. We can observe from this figure that the average power splitting ratio of the AF networks is always larger than that of the DF networks. Furthermore, the average power splitting ratios of both AF and DF networks experience a steady increase when the number of links grows. A similar phenomenon can be observed from Fig. \ref{fig:impact_SNR} (b), in which the curves for the average power splitting ratios of both AF and DF networks are plotted versus the sources' transmit powers. This is because increasing the number of links and rising the transmit powers of all links achieve the same effect as increased mutual interference. The impact of the transmit powers of all links on the average sum-rates of both AF and DF networks is shown in Fig. \ref{fig:impact_SNR} (a). As can be observed from this figure, the average sum-rate achieved by the proposed game-theoretical approach steadily increases from a low to a moderate SNR region and tends to get saturated at high SNR, above 20 (dB). This indicates that at high SNR, the power control at the sources should be jointly considered with the power splitting at the relays to further improve the sum-rate, which will be considered in our future work. Finally, we can see from Figs. \ref{fig:impact_link_num} (a) and \ref{fig:impact_SNR} (a) that a DF network can achieve a higher average sum-rate than an AF network with the same settings.

\begin{figure}
\centering \scalebox{0.55}{\includegraphics{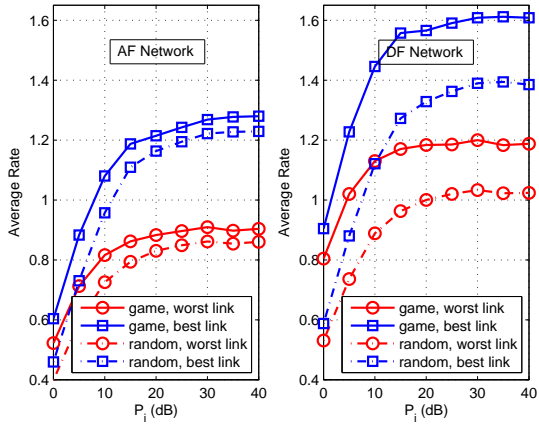}}
\caption{The average rate of the best and worst links in AF and DF networks with $d_{\max} = 5$ and $N = 5$.\label{fig:best_worst_link}}
\end{figure}

In Fig. \ref{fig:best_worst_link}, we plot the average rate curves for the best and worst links in both AF and DF networks with the same network setting as in Fig. \ref{fig:impact_SNR}. Similar phenomena shown in Fig. \ref{fig:impact_SNR} (a) can also be observed in this figure. Besides, we can see from Fig. \ref{fig:best_worst_link} that compared with the random scheme, the proposed game-theoretical scheme can effectively improve the average rates of the best and worst links in both AF and DF networks. Furthermore, this performance improvement in the DF network is shown to be more significant than that in the AF network.

\section{Conclusions}
In this paper, we developed a game-theoretical framework to address the distributed power splitting problem for simultaneous wireless information and power transfer (SWIPT) in relay interference channels. We formulated non-cooperative games for three different network scenarios, in which each link is modeled as a strategic player who aims to maximize its individual achievable rate by choosing the dedicated relay's power splitting ratio. We showed that the formulated games always achieve a unique Nash equilibrium (NE). Best response-based distributed algorithms with provable convergence were also developed to achieve the NEs. The numerical results showed that the proposed algorithms can converge to the corresponding NEs from different starting points, and the developed game-theoretical approach can achieve a near-optimal network-wide performance on average, especially for the scenarios with relatively low and moderate interference.
\begin{appendix}
\subsection{Proof of Proposition \ref{prop:NE_existence_AF_game}}\label{append:prop1}
Firstly, it is straightforward to observe that the utility function ${u_i^{AF}}\left( {{\rho _i},{\pmb \rho _{ - i}}} \right)$ is continuous in $\rho_i$. Then, a feasible method to prove the quasi-concavity of the utility function ${u_i^{AF}}\left( {{\rho _i},{\pmb \rho _{ - i}}} \right)$ is to use the following theorem  \cite[pp. 99]{Boyd_Convex_2004}:
\begin{theorem}\label{thm:quasi-concave}
A continuous function $f:\mathbb{R} \to \mathbb{R}$ is quasi-concave if and only if at least one of the following conditions holds:
\begin{itemize}
  \item $f$ is non-decreasing
  \item $f$ is non-increasing
  \item there is a point $c \in {\bf{dom}}\;f$ such that for $t\le c$ (and $t \in {\bf{dom}}\;f$), $f$ is non-decreasing, and for $t \ge c$ (and $t \in {\bf{dom}}\;f$), $f$ is non-increasing.
\end{itemize}
\end{theorem}

Next, let us prove the quasi-concavity of the utility function by showing that it is first non-decreasing and then non-increasing in the feasible domain. This will be achieved by demonstrating that the first-order derivative of the utility is no less than zero when the variable is smaller than a certain value and is no larger than zero in the remaining domain.

To proceed, we derive the first-order derivative of the function ${u_i^{AF}}\left( {{\rho _i},{\pmb \rho _{ - i}}} \right)$ with respective to (w.r.t) $\rho_i$, which, after some algebra manipulations, can be expressed as
\begin{equation}\label{eq:derivative_AF_utility}
\begin{split}
&\frac{{\partial u_i^{AF}\left( {{\rho _i},{{\pmb \rho} _{ - i}}} \right)}}{{\partial {\rho _i}}} = \\
&\frac{1}{{\ln 2}}\frac{{{C_i}{{\left( {{\rho _i}} \right)}^2} - 2{D_i}{\rho _i} + {D_i}}}{{{{\left[ {{\rho _i}\left( {1 - {\rho _i}} \right){Y_i}{Z_i} - {C_i}{\rho _i} + {D_i}} \right]}^2}/{X_i}{Z_i} + {\rho _i}\left( {1 - {\rho _i}} \right)}},
\end{split}
\end{equation}
where
\begin{subequations}\label{eq:C_i_D_i}
\begin{align}
 &{C_i} = \left( {{X_i} + {Y_i}} \right)\left( {{W_i} + 1} \right) - {Z_i}, \\
 &{D_i} = \left( {{X_i} + {Y_i} + 1} \right)\left( {{W_i} + 1} \right),
\end{align}
\end{subequations}
are defined for the simplicity of notations.

After a careful observation on the right-hand side (RHS) of (\ref{eq:derivative_AF_utility}), we can deduce that the sign of $\frac{{\partial u_i^{AF}\left( {{\rho _i},{{\pmb \rho} _{ - i}}} \right)}}{{\partial {\rho _i}}}$ is only determined by the numerator
\begin{equation}\label{eq:func_kappa}
{\kappa _i}\left( {{\rho _i}} \right) = {C_i}{\left( {{\rho _i}} \right)^2} - 2{D_i}{\rho _i} + {D_i},
\end{equation}
since the denominator is always large than zero. To further determine the monotonicity of the function ${u_i^{AF}}\left( {{\rho _i},{\pmb \rho _{ - i}}} \right)$, we need to investigate the properties of the quadratic function ${\kappa _i}\left( {{\rho _i}} \right)$ on the feasible set of $\rho_i$ (i.e., $[0,1]$). Firstly, we note that
\begin{subequations}
\begin{align}
&{\kappa _i}\left( {{0}} \right) = D_i > 0,\\
&{\kappa _i}\left( {{1}} \right) = C_i - D_i = -(Z_i +W_i+1) <0.
\end{align}
\end{subequations}
Then, we can draw all possible shapes of the function ${\kappa _i}\left( {{\rho _i}} \right)$ versus $\rho_i$ for different cases, as depicted in Fig. \ref{fig_kappa_func}. From Fig. \ref{fig_kappa_func}, we can see that in spite of the sign of the term $C_i$, there always exists a point $\epsilon_i \in [0,1]$ such that ${\kappa _i}\left( {{\epsilon_i}} \right) = 0$, ${\kappa _i}\left( {{\rho_i}} \right) > 0$ for $\rho_i < \epsilon_i$, and ${\kappa _i}\left( {{\rho_i}} \right) < 0$ for $\rho_i > \epsilon_i$. This means that on its feasible domain, the function ${u_i^{AF}}\left( {{\rho _i},{\pmb \rho _{ - i}}} \right)$ is increasing in $\rho_i$ before the point $\epsilon_i$ and is decreasing in $\rho_i$ after the point $\epsilon_i$. Then, according to Theorem \ref{thm:quasi-concave}, we can claim that the utility function ${u_i^{AF}}\left( {{\rho _i},{\pmb \rho _{ - i}}} \right)$ is quasi-concave in $\rho_i$.

\begin{figure*}
\centering
 \subfigure[When $C_i > 0$]
  {\scalebox{0.8}{\includegraphics {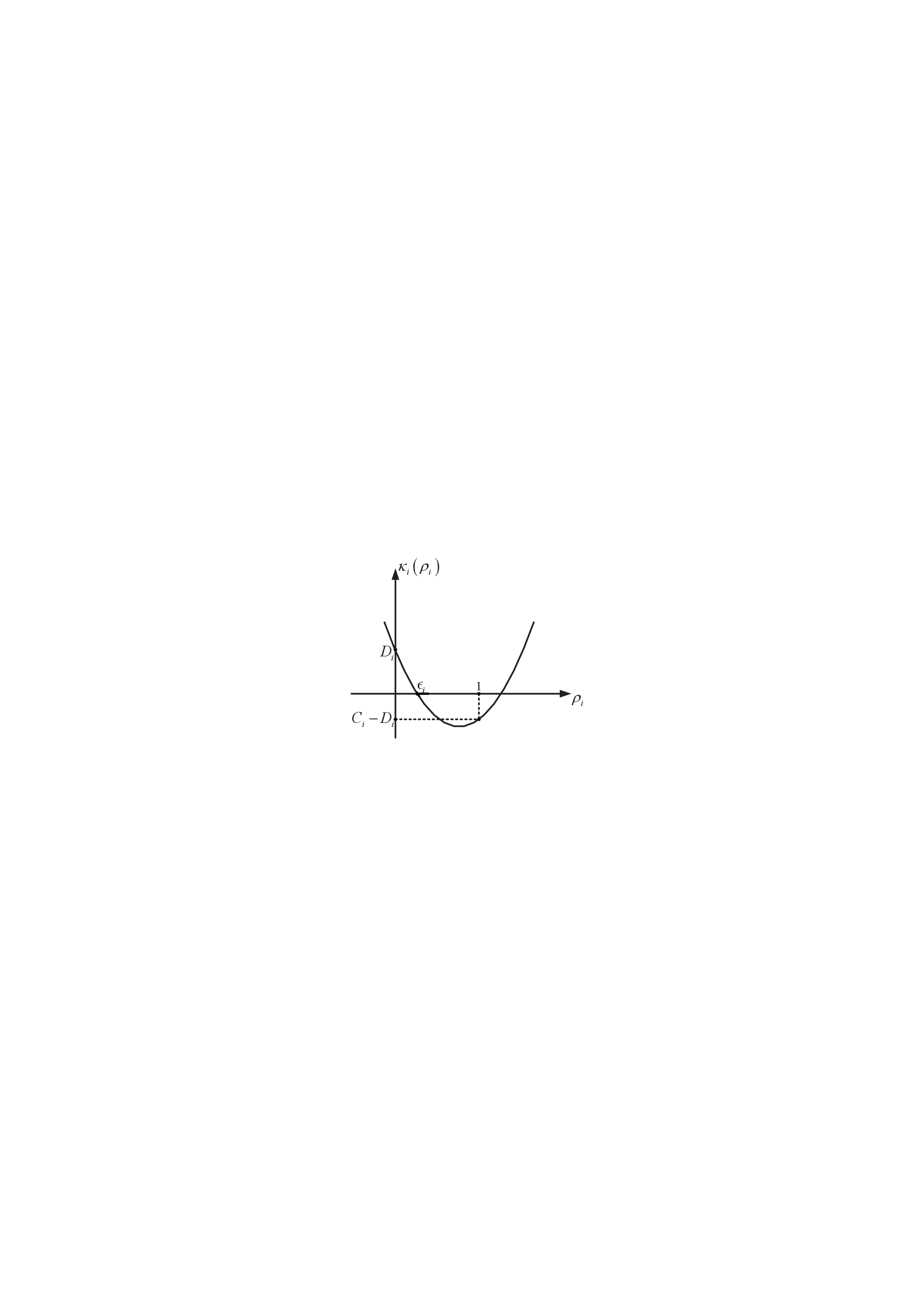}
  \label{fig_kappa_func_a}}}
\hfil
 \subfigure[When $C_i < 0$]
  {\scalebox{0.8}{\includegraphics {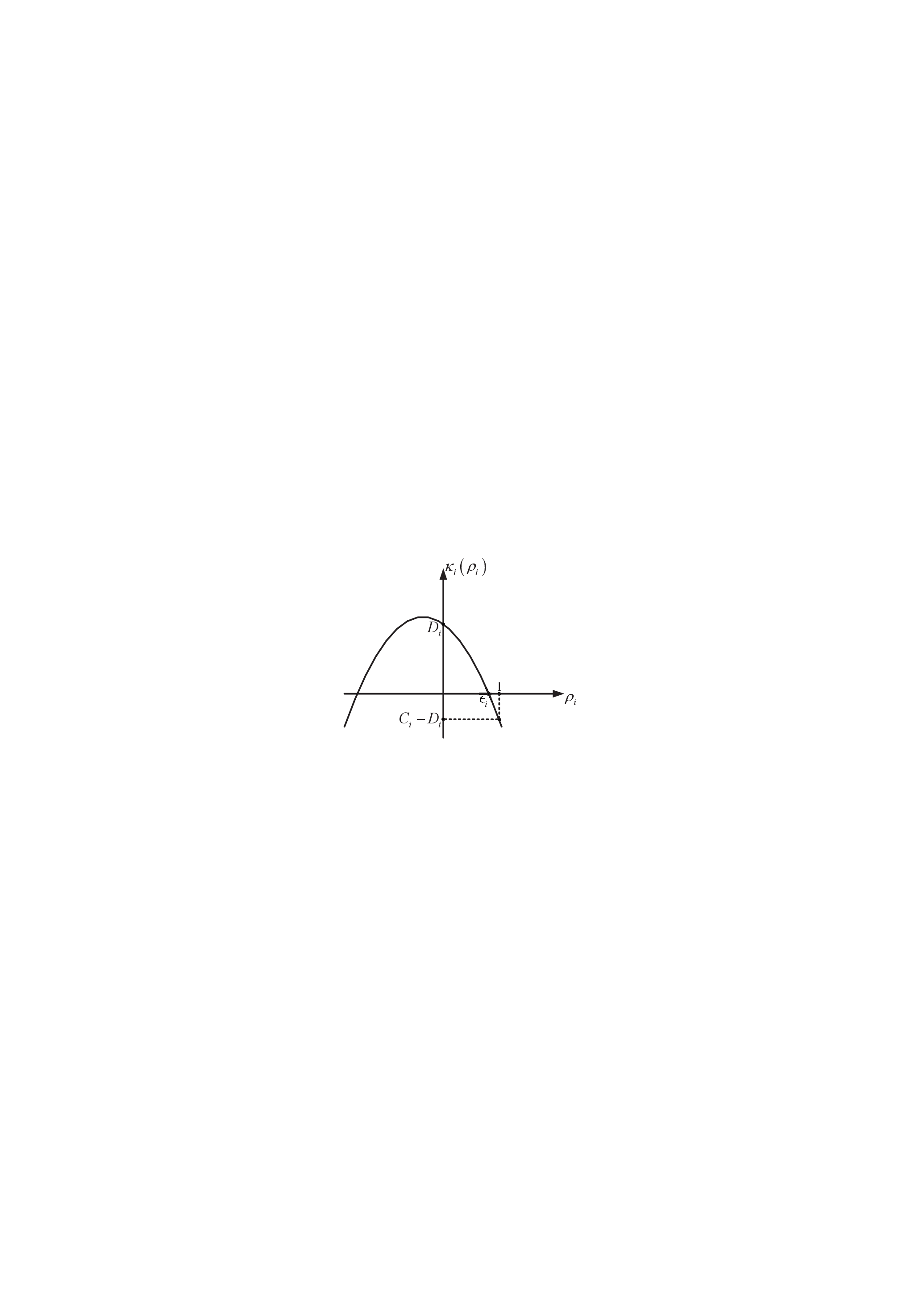}
\label{fig_kappa_func_b}}}
\hfil
 \subfigure[When $C_i = 0$]
  {\scalebox{0.8}{\includegraphics {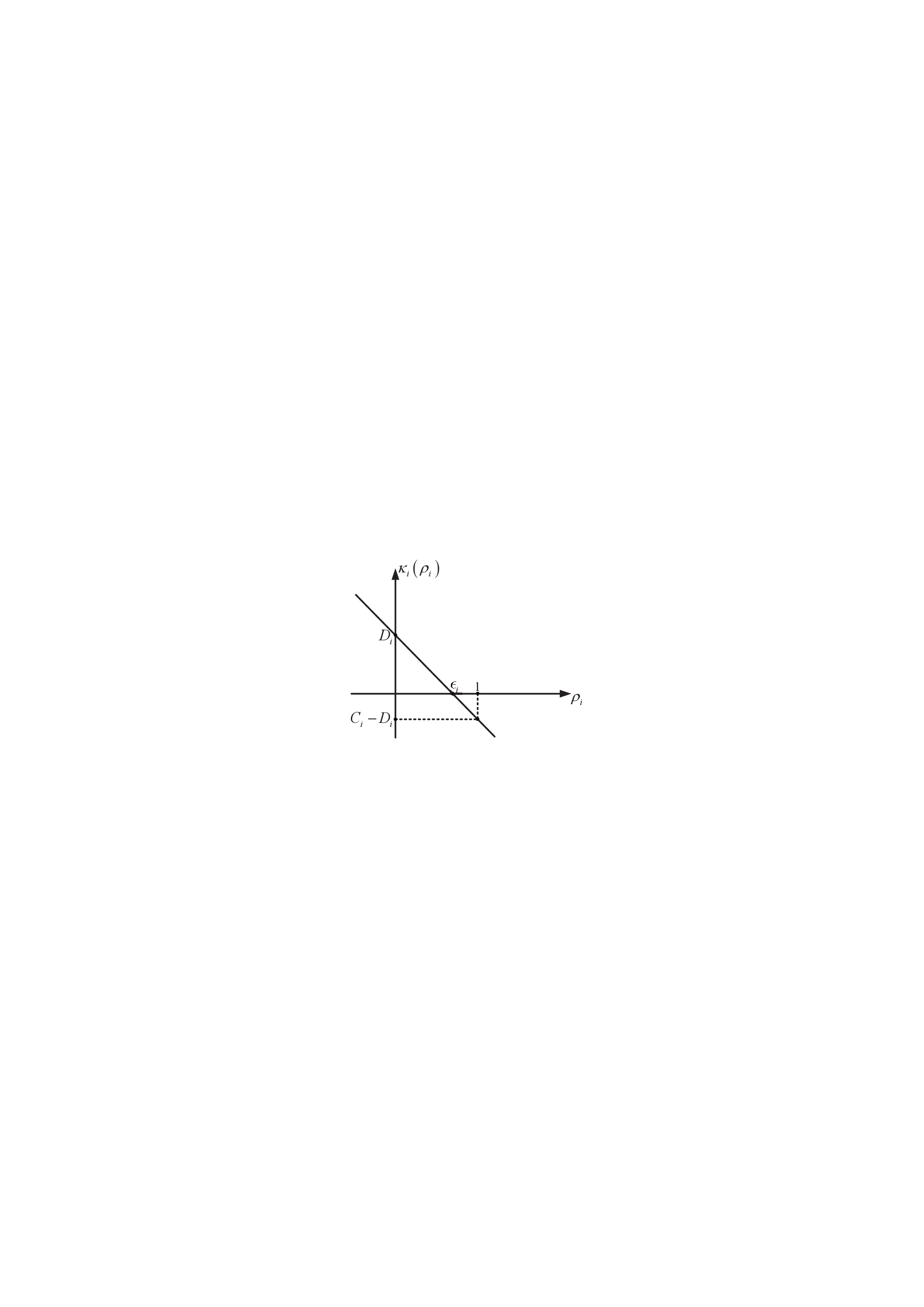}
\label{fig_kappa_func_c}}}
\caption{All possible shapes for the function ${\kappa _i}\left( {{\rho _i}} \right)$ versus $\rho_i$ for different cases.}
\label{fig_kappa_func}
\end{figure*}

In addition, it is easy to check that the feasible set ${\mathcal A}_i$ is compact and convex and the utility function ${u_i^{AF}}\left( {{\rho _i},{\pmb \rho _{ - i}}} \right)$ is continuous in $\pmb \rho$ for any $i \in {\mathcal N}$. Then, with reference to Theorem \ref{thm:existence_NE}, we can further claim that the formulated game ${\mathcal{G}}_{AF}$ for the pure AF network possesses at least one NE.

For the formulated game ${\mathcal{G}}_{DF}$, it is straightforward to check that both $\gamma _{i,1}^{DF}$ and $\gamma _{i,2}^{DF}$ are concave functions of $\rho_i$. Moreover, the minimum of two concave functions is also concave, which means that the SINR $\gamma _i^{DF}$ is a concave function of the strategy $\rho_i$. Since the function $f\left(x\right) = \log\left(1 + x\right)$ is concave and non-decreasing, we can claim\footnote{The composite function $h\left(x\right) = f\left(g\left(x\right)\right)$ is concave in $x$ if $f\left(x\right)$ is concave and non-decreasing, and $g\left(x\right)$ is concave \cite{Boyd_Convex_2004}.} that the utility function ${u_i^{DF}}\left( {{\rho _i},{\pmb \rho _{ - i}}} \right)$ is concave (quasi-concave) in $\rho_i$. Thus, according to Theorem \ref{thm:existence_NE}, we can conclude that the formulated game ${\mathcal{G}}_{DF}$ also admits at least one NE. This completes the proof.

\subsection{Proof of Lemma \ref{lemma:BR_AF}}\label{append:lemma_BR_AF}
According to the analysis in Appendix \ref{append:prop1}, for a given power splitting profile $\pmb \rho$, the expression of the parameter $\epsilon_i$ is actually the best response function of the $i$th link. Moreover, $\epsilon_i$ is one solution of the following quadratic equation:
\begin{equation}\label{eq:qudratic_func_epsilon}
{\kappa _i}\left( {{\epsilon _i}} \right) = {C_i}{\left( {{\epsilon _i}} \right)^2} - 2{D_i}{\epsilon _i} + {D_i} = 0.
\end{equation}
According to the value of $C_i$, we divide the derivation for the expression of $\epsilon_i$ to three cases:

\emph{(a) When $C_i = 0$}: The equation in (\ref{eq:qudratic_func_epsilon}) is simplified as
\begin{equation}
 - 2{D_i}{\epsilon _i} + {D_i} = 0.
\end{equation}
Thus, $\epsilon _i = 1/2$ when $C_i = 0$.

\emph{(b) When $C_i > 0$}: In this case, the quadratic equation in (\ref{eq:qudratic_func_epsilon}) possesses the following two roots:
\begin{equation}
\begin{split}
\epsilon_{i,1} &= \frac{{{D_i} + \sqrt {{{\left( {{D_i}} \right)}^2} - {C_i}{D_i}} }}{{{C_i}}}\\
 &= \frac{{\sqrt {{D_i}} \left( {\sqrt {{D_i}}  + \sqrt {{D_i} - {C_i}} } \right)}}{{{D_i} - \left( {{D_i} - {C_i}} \right)}}\\
 &= \frac{{\sqrt {{D_i}} }}{{\sqrt {{D_i}}  - \sqrt {{D_i} - {C_i}} }},
 \end{split}
\end{equation}
\begin{equation}
\begin{split}
{\epsilon_{i,2}} &= \frac{{{D_i} - \sqrt {{{\left( {{D_i}} \right)}^2} - {C_i}{D_i}} }}{{{C_i}}}\\
 &= \frac{{\sqrt {{D_i}} }}{{\sqrt {{D_i}}  + \sqrt {{D_i} - {C_i}} }}.
 \end{split}
\end{equation}
Since $C_i>0$ and $D_i>D_i-C_i>0$, we have $\epsilon_{i,1}>1$ and $\frac{1}{2}<\epsilon_{i,2}<1$. Thus, the valid expression of $\epsilon_i$ when $C_i > 0$ is given by
\begin{equation}
\begin{split}
\epsilon_i = {\epsilon_{i,2}} = \frac{{\sqrt {{D_i}} }}{{\sqrt {{D_i}}  + \sqrt {{D_i} - {C_i}} }},
 \end{split}
\end{equation}
because the feasible set of $\epsilon_i$ is $[0,1]$.

\emph{(c) When $C_i < 0$}: In this case, we have $D_i-C_i > D_i >0$. Thus, the roots $\epsilon_{i,1}<0$ and $0<\epsilon_{i,2}<\frac{1}{2}$. Hence, the valid expression of $\epsilon_i$ when $C_i < 0$ is the same with the case when $C_i >0 $, i.e.,
\begin{equation}
\begin{split}
\epsilon_i = {\epsilon_{i,2}} = \frac{{\sqrt {{D_i}} }}{{\sqrt {{D_i}}  + \sqrt {{D_i} - {C_i}} }}.
 \end{split}
\end{equation}

Therefore, substituting the expressions of $C_i$ and $D_i$ defined in (\ref{eq:C_i_D_i}) and setting ${\mathcal B}_i^{AF}\left(\pmb \rho  \right) = \epsilon_i$, we can obtain the desired result in (\ref{eq:BR_function_AF}) for the best response function of the $i$th link in pure AF network. This completes the proof.

\subsection{Proof of Proposition \ref{prop:NE_unique_AF}}\label{append:Prop_NE_unique_AF}

To prove this proposition, we need to show that the function ${\bf{\mathcal B}}^{AF}\left( \pmb \rho  \right)$ for the formulated game is standard. Firstly, based on the analysis in Appendix \ref{append:lemma_BR_AF}, it is easy to verify that the best response function satisfies the conditions in Definition \ref{eq:standard_func} when $\left( {{X_i} + {Y_i}} \right)\left( {{W_i} + 1} \right) = {Z_i}$ (i.e., $C_i = 0$). Thus, we only need to show the function
\begin{equation}\label{eq:BR_AF_2}
{\bf{\mathcal B}}_i^{AF}\left( \pmb \rho  \right) = \frac{{\sqrt {{\left( {{X_i} + {Y_i} + 1} \right)\left( {{W_i} + 1} \right)}} }}{{\sqrt {\left( {{X_i} + {Y_i} + 1} \right)\left( {{W_i} + 1} \right)}  + \sqrt {Z_i +W_i+1} }}
\end{equation}
also satisfies the three properties of the standard function for any $i$, which are proved in the following:

\emph{(1) Positivity}: According to the analysis in Appendix \ref{append:lemma_BR_AF}, for any player $i$ and any strategy profile $\pmb \rho$, the best response function ${\bf{\mathcal B}}_i^{AF}\left( \pmb \rho  \right) $ is always larger than 0, which guarantees the positivity of the function ${\bf{\mathcal B}}^{AF}\left( \pmb \rho  \right)$.

\emph{(2) Monotonicity}: Recall the definition of the terms $X_i$, $Y_i$, $Z_i$ and $W_i$ in (\ref{eq:def_XYZW}). We can see that only the term $W_i$ is related to the strategy profile $\pmb \rho$. Suppose $\pmb \rho$ and ${\pmb \rho}^\prime$ are two different strategy profiles and $\pmb \rho \ge {\pmb \rho}^\prime$. Then, the corresponding best response functions of any player $i$ can be written as
\begin{equation}\label{eq:BR_AF_rho}
\begin{split}
{{\mathcal B}_i^{AF}}\left(\pmb \rho  \right) = \frac{{\sqrt {{\left( {{X_i} + {Y_i} + 1} \right)\left( {{W_i} + 1} \right)}} }}{{\sqrt {\left( {{X_i} + {Y_i} + 1} \right)\left( {{W_i} + 1} \right)}  + \sqrt {Z_i +W_i+1} }},
\end{split}
\end{equation}
and
\begin{equation}\label{eq:BR_AF_rho_prime}
\begin{split}
{{\mathcal B}_i^{AF}}\left({\pmb \rho}^\prime  \right) =\frac{{\sqrt {{\left( {{X_i} + {Y_i} + 1} \right)\left( {{W_i^\prime} + 1} \right)}} }}{{\sqrt {\left( {{X_i} + {Y_i} + 1} \right)\left( {{W_i^\prime} + 1} \right)}  + \sqrt {Z_i +W_i^\prime+1} }},
\end{split}
\end{equation}
where ${W_i^\prime} = \sum\nolimits_{j = 1,j \ne i}^N {{\rho _j^\prime}\eta \left( {\sum\nolimits_{n = 1}^N {{P_n}{{\left| {{g_{nj}}} \right|}^2}} } \right){{\left| {{h_{ji}}} \right|}^2}}/ {\sigma ^2}$.

After a careful comparison of (\ref{eq:BR_AF_rho}) and (\ref{eq:BR_AF_rho_prime}), we can see that all the terms in ${{\mathcal B}_i^{AF}}\left(\pmb \rho  \right)$ and ${{\mathcal B}_i^{AF}}\left({\pmb \rho}^\prime  \right)$ are the same except ${W_i}$ and ${W_i^\prime}$. Hence, the inequity ${{\mathcal B}_i^{AF}}\left(\pmb \rho  \right) \ge{{\mathcal B}_i^{AF}}\left({\pmb \rho}^\prime  \right)$ holds if we can prove ${{\mathcal B}_i^{AF}}\left(W_i \right) \ge{{\mathcal B}_i^{AF}}\left({{W_i^\prime}} \right)$. Moreover, we have ${W_i} \ge {W_i^\prime}$ since $\pmb \rho \ge {\pmb \rho}^\prime$. Thus, the proof of ${{\mathcal B}_i^{AF}}\left(\pmb \rho  \right) \ge{{\mathcal B}_i^{AF}}\left({\pmb \rho}^\prime  \right)$ is equivalent to proving that the ${{\mathcal B}_i^{AF}}\left(W_i \right)$ is non-decreasing in $W_i$. To proceed, we re-write the best response function ${{\mathcal B}_i^{AF}}\left(\pmb \rho  \right)$ as
\begin{equation}\label{eq:BR_AF_rho_rewrite}
\begin{split}
{{\mathcal B}_i^{AF}}\left(\pmb \rho  \right) = \frac{1}{{1 + \sqrt {\frac{{{Z_i}}}{{\left( {{X_i} + {Y_i} + 1} \right)\left( {{W_i} + 1} \right)}} + \frac{1}{{{X_i} + {Y_i} + 1}}} }},
\end{split}
\end{equation}
From (\ref{eq:BR_AF_rho_rewrite}), we can easily observe that ${{\mathcal B}_i^{AF}}\left(\pmb \rho  \right)$ is increasing in $W_i$, which complete the proof of monotonicity.

\emph{(3) Scalability}: For any given $\alpha > 1$, we have
\begin{equation}\label{}
\begin{split}
\alpha{{\mathcal B}_i^{AF}}\left(\pmb \rho  \right) &= \frac{\alpha}{{1 + \sqrt {\frac{{{Z_i}}}{{\left( {{X_i} + {Y_i} + 1} \right)\left( {{W_i} + 1} \right)}} + \frac{1}{{{X_i} + {Y_i} + 1}}} }}\\
&= \frac{1}{{\frac{1}{\alpha } + \sqrt {\frac{{{Z_i}}}{{{\alpha ^2}\left( {{X_i} + {Y_i} + 1} \right)\left( {{W_i} + 1} \right)}} + \frac{1}{{{\alpha ^2}\left( {{X_i} + {Y_i} + 1} \right)}}} }} \\
&> \frac{1}{{1 + \sqrt {\frac{{{Z_i}}}{{\left( {{X_i} + {Y_i} + 1} \right)\left( {\alpha {W_i} + 1} \right)}} + \frac{1}{{\left( {{X_i} + {Y_i} + 1} \right)}}} }}\\
&={{\mathcal B}_i^{AF}}\left(\alpha\pmb \rho  \right),
\end{split}
\end{equation}
which proves that the best response function ${{\mathcal B}_i^{AF}}\left(\pmb \rho  \right)$ meets the scalability property. This completes the proof.

\subsection{Proof of Lemma \ref{lemma:BR}}\label{append:lemma_BR_function_DF}

Let $\rho_i^*$ denote the best response of the $i$th link corresponding to a given power splitting strategy profile $\pmb \rho$. That is, $\rho_i^*$ is the maximizer of $u_i^{DF}\left(\pmb \rho\right)$. It can be easily observed that $\gamma _{i,1}^{DF}$ in (\ref{eq:DF_SNR_1st_hop}) is monotonically decreasing in $\rho_i$, while $\gamma _{i,2}^{DF}$ in (\ref{eq:DF_SNR_2rd_hop}) is monotonically increasing in $\rho_i$. Thus, $\rho_i^*$ must satisfy the following condition:
\begin{equation}
\gamma _{i,1}^{DF}\left(\rho_i^*\right) = \gamma _{i,2}^{DF}\left(\rho_i^*\right).
\end{equation}
Substituting the expression of $\gamma _{i,1}^{DF}$ and $\gamma _{i,2}^{DF}$, we have
\begin{equation}\label{eq:equality}
\begin{split}
\frac{{\left( {1 - {\rho _i^*}} \right){X_i}}}{{\left( {1 - {\rho _i^*}} \right){Y_i} + 1}} = \frac{{{\rho _i^*}{Z_i}}}{{{W_i} + 1}}.
\end{split}
\end{equation}
After rearranging (\ref{eq:equality}), we obtain
\begin{equation}\label{eq:quadratic_equality}
\begin{split}
&{Y_i}{Z_i}{\left( {\rho _i^*} \right)^2} - \left[ {{X_i}\left( {{W_i} + 1} \right) + {Y_i}{Z_i} + {Z_i}} \right]\rho _i^* \\
&\;\;\;\;\;\;\;\;\;\;\;\;\;\;\;\;\;\;\;\;+ {X_i}\left( { W_i  + 1} \right)= \ell_i
\left(\rho _i^*\right) = 0.
\end{split}
\end{equation}
Note that (\ref{eq:quadratic_equality}) is a quadratic equality of ${\rho _i^*}$, denoted by $\ell_i\left(\rho _i^*\right) = 0$. Thus, the value of $\rho _i^*$ can be obtained by solving the quadratic equality in (\ref{eq:quadratic_equality}).

We note that
\begin{subequations}\label{eq:properity_quadratic_equality_DF}
\begin{align}
&{Y_i}{Z_i}>0,\\
&\ell_i\left(0\right) = {X_i}\left( {W_i  + 1} \right) >0,\\
&\ell_i\left(1\right) = -{Z_i} <0.
\end{align}
\end{subequations}
Based on (\ref{eq:properity_quadratic_equality_DF}), we can deduce that the quadratic equality $\ell_i\left(\rho _i^*\right) = 0$ admits two roots lying in $(0,1)$ and $(1,+\infty)$. Since the feasible set of the power splitting ratio is $[0,1]$, the valid value of $\rho _i^*$ can only be the smaller root of the equality $\ell_i\left(\rho _i^*\right) = 0$. Mathematically, we have,
\begin{equation}\label{eq:BR_function_appendix}
\begin{split}
\rho_i^* &= \left[ {\left( {{X_i}{W_i} + {X_i} + {Y_i}{Z_i} + {Z_i}} \right) - } \right. \\
&\left. {\sqrt {{{\left( {{X_i}{W_i} + {X_i} - {Y_i}{Z_i} + {Z_i}} \right)}^2} + 4{Y_i}Z_i^2} } \right]/\left( {2{Y_i}{Z_i}} \right)\\
&={{\mathcal B}_i^{DF}}\left(\pmb \rho  \right),
\end{split}
\end{equation}
which completes the proof.

\subsection{Proof of Proposition \ref{prop:NE_unique}}\label{append:Prop_2}

Similar to the proof of Proposition \ref{prop:NE_unique_AF} in Appendix \ref{append:Prop_NE_unique_AF}, the proof of this proposition follows by showing that the function ${\pmb{\mathcal B}^{DF}}\left( \pmb \rho  \right)$ of the formulated game is standard. In the following, we show that the function ${\pmb{\mathcal B}^{DF}}\left( \pmb \rho  \right)$ satisfies the three properties of the standard function.

\emph{(1) Positivity}: As shown in Appendix \ref{append:lemma_BR_function_DF}, for any player $i$ and any strategy profile $\pmb \rho$, the best response function ${{\mathcal B}}_i^{DF}\left( \pmb \rho  \right) $ is always larger than 0, which guarantees the positivity of the function ${\pmb{\mathcal B}^{DF}}\left( \pmb \rho  \right)$.

\emph{(2) Monotonicity}: Suppose $\pmb \rho$ and ${\pmb \rho}^\prime$ are two different strategy profiles and $\pmb \rho \ge {\pmb \rho}^\prime$. Then, the corresponding best response functions of any player $i$ can be written as
\begin{equation}\label{eq:app_1}
\begin{split}
{{\mathcal B}_i^{DF}}\left(\pmb \rho  \right)&= \left[ {\left( {{X_i}{W_i} + {X_i} + {Y_i}{Z_i} + {Z_i}} \right) - } \right. \\
&\left. {\sqrt {{{\left( {{X_i}{W_i} + {X_i} - {Y_i}{Z_i} + {Z_i}} \right)}^2} + 4{Y_i}Z_i^2} } \right]/\left( {2{Y_i}{Z_i}} \right),
\end{split}
\end{equation}
and
\begin{equation}\label{eq:app_2}
\begin{split}
{{\mathcal B}_i^{DF}}\left({\pmb \rho}^\prime  \right)&= \left[ {\left( {{X_i}{W_i^\prime} + {X_i} + {Y_i}{Z_i} + {Z_i}} \right) - } \right. \\
&\left. {\sqrt {{{\left( {{X_i}{W_i^\prime} + {X_i} - {Y_i}{Z_i} + {Z_i}} \right)}^2} + 4{Y_i}Z_i^2} } \right]/\left( {2{Y_i}{Z_i}} \right),
\end{split}
\end{equation}
where $X_i$, $Y_i$, $Z_i$, $W_i$ are defined in (\ref{eq:def_XYZW}), and ${W_i^\prime} = \sum\nolimits_{j = 1,j \ne i}^N {{\rho _j^\prime}\eta \left( {\sum\nolimits_{n = 1}^N {{P_n}{{\left| {{g_{nj}}} \right|}^2}} } \right){{\left| {{h_{ji}}} \right|}^2}}/ {\sigma ^2}$.

Analogous to the analyses in Appendix \ref{append:Prop_NE_unique_AF}, the proof of ${{\mathcal B}_i^{DF}}\left(\pmb \rho  \right) \ge{{\mathcal B}_i^{DF}}\left({\pmb \rho}^\prime  \right)$ is equivalent to proving that ${\textstyle{{\partial {{\mathcal B}_i^{DF}}\left( {{W_i}} \right)} \over {{W_i}}}} \ge 0$. Expanding ${\textstyle{{\partial {{\mathcal B}_i^{DF}}\left( {{W_i}} \right)} \over {{W_i}}}} \ge 0$, we have
\begin{equation}\label{eq:app_3}
\begin{split}
&\frac{{\partial {{\mathcal B}_i^{DF}}\left( {{W_i}} \right)}}{{{W_i}}} \\
&= \frac{{{X_i}}}{{2{Y_i}{Z_i}}}\left[ {1 - \frac{{{X_i}{W_i} + {X_i} - {Y_i}{Z_i} + {Z_i}}}{{\sqrt {{{\left( {{X_i}{W_i} + {X_i} - {Y_i}{Z_i} + {Z_i}} \right)}^2} + 4{Y_i}Z_i^2} }}} \right].
\end{split}
\end{equation}
Since the term ${{{X_i}}}/{{\left(2{Y_i}{Z_i}\right)}} >0$ and the term in the square bracket of (\ref{eq:app_3}) is always large than zero, we can claim that ${\textstyle{{\partial {{\mathcal B}_i^{DF}}\left( {{W_i}} \right)} \over {{W_i}}}} > 0$, which complete the proof of monotonicity.

\emph{(3) Scalability}: For any $\alpha >1$, we define the function $ {\mathcal F}_i\left( {\alpha ,{\pmb \rho} } \right) =  \alpha {{\mathcal B}}_i^{DF}\left( \pmb \rho  \right) - {{\mathcal B}}_i^{DF}\left( \alpha {\pmb \rho}  \right)$. Then, the proof of the scalability is equivalent to proving that ${\mathcal F}_i\left( {\alpha ,{\pmb \rho} } \right) > 0$ for any $\alpha >1$. Firstly, it is obvious that $ {\mathcal F}_i\left( {1 ,{\pmb \rho} } \right) = 0$. Thus, a sufficient condition for ${\mathcal F}_i\left( {\alpha ,{\pmb \rho} } \right) > 0$ is that ${\mathcal F}_i\left( {\alpha ,{\pmb \rho} } \right)$ is an increasing function of $\alpha$, i.e., ${\textstyle{{\partial {{\cal F}_i}\left( {\alpha ,{ \pmb \rho} } \right)} \over {\partial \alpha }}} > 0$. To proceed, we first derive the first-order and second-order partial derivatives of ${\mathcal F}_i\left( {\alpha ,{\pmb \rho} } \right)$ w.r.t $\alpha$ and obtain
\begin{equation}\label{eq:app_4_1}
\begin{split}
\frac{{\partial {{\cal F}_i}\left( {\alpha ,{\pmb \rho} } \right)}}{{\partial \alpha }} &= \frac{1}{{2{Y_i}{Z_i}}}\left\{ {{X_i} + {Y_i}{Z_i} + {Z_i}} \right.  \\
 &{ + \frac{{\left( {\alpha {X_i}{W_i} + {X_i} - {Y_i}{Z_i} + {Z_i}} \right){X_i}{W_i}}}{{\sqrt {{{\left( {\alpha {X_i}{W_i} + {X_i} - {Y_i}{Z_i} + {Z_i}} \right)}^2} + 4{Y_i}{Z_i}^2} }}} \\
&\left. - \sqrt {{{\left( {{X_i}{W_i} + {X_i} - {Y_i}{Z_i} + {Z_i}} \right)}^2} + 4{Y_i}Z_i^2} \right\},
\end{split}
\end{equation}
\begin{equation}\label{eq:app_4}
\begin{split}
\frac{{{\partial ^2}{{\cal F}_i}\left( {\alpha ,{\pmb \rho} } \right)}}{{\partial {\alpha ^2}}} &= \frac{{2{Z_i}{{\left( {{X_i}{W_i}} \right)}^2}}}{\left[{{{\left( {\alpha{X_i}{W_i} + {X_i} - {Y_i}{Z_i} + {Z_i}} \right)}^2} + 4{Y_i}Z_i^2}\right]^{3/2}} \\
%&> 0.
\end{split}
\end{equation}
From (\ref{eq:app_4}), we can see that $\frac{{{\partial ^2}{{\cal F}_i}\left( {\alpha ,{\pmb \rho} } \right)}}{{\partial {\alpha ^2}}}$ is always larger than 0, which indicates that ${\textstyle{{\partial {{\cal F}_i}\left( {\alpha ,{ \pmb \rho} } \right)} \over {\partial \alpha }}}$ is increasing in $\alpha$. Thus, a sufficient condition for ${\mathcal F}_i\left( {\alpha ,{\pmb \rho} } \right) > 0$ can now be simplified as ${\left. {{\textstyle{{\partial {{\cal F}_i}\left( {\alpha ,{\pmb \rho} } \right)} \over {\partial \alpha }}}} \right|_{\alpha  = 1}} > 0$. Substituting $\alpha = 1$ into (\ref{eq:app_4_1}), we get
\begin{equation}\label{eq:app_5}
\begin{split}
&{\left. {\frac{{\partial {{\cal F}_i}\left( {\alpha ,{\pmb \rho} } \right)}}{{\partial \alpha }}} \right|_{\alpha {\rm{ = }}1}} = \frac{1}{{2{Y_i}{Z_i}}} \left\{ {{X_i} + {Y_i}{Z_i} + {Z_i}} \right.\\
 &~~~~~~~~~~+ \frac{{\left( {{X_i}{W_i} + {X_i} - {Y_i}{Z_i} + {Z_i}} \right){X_i}{W_i}}}{{\sqrt {{{\left( {{X_i}{W_i} + {X_i} - {Y_i}{Z_i} + {Z_i}} \right)}^2} + 4{Y_i}Z_i^2} }}\\
 &~~~~~~~~~~\left.- \sqrt {{{\left( {{X_i}{W_i} + {X_i} - {Y_i}{Z_i} + {Z_i}} \right)}^2} + 4{Y_i}Z_i^2}\right\}.
\end{split}
\end{equation}

To proceed, we derive the first-order derivative for the RHS of (\ref{eq:app_5}) with respect to $W_i$. After some algebraic manipulations, we obtain
\begin{equation*}
\begin{split}
&\partial \left( {{{\left. {{\textstyle{{\partial {{\cal F}_i}\left( {\alpha ,{\pmb \rho} } \right)} \over {\partial \alpha }}}} \right|}_{\alpha  = 1}}} \right)/\partial {W_i} \\
&= \frac{{2{Z_i}{W_i}X_i^2}}{{{{\left[ {{{\left( {{X_i}{W_i} + {X_i} - {Y_i}{Z_i} + {Z_i}} \right)}^2} + 4{Y_i}Z_i^2} \right]}^{3/2}}}},
\end{split}
\end{equation*}
which is shown to be always positive. Thus, ${\left. {\frac{{\partial {{\cal F}_i}\left( {\alpha ,{\pmb \rho} } \right)}}{{\partial \alpha }}} \right|_{\alpha {\rm{ = }}1}}$ is an increasing function in $W_i$. Since $W_i > 0$, we further have
\begin{equation}\label{eq:app_6}
\begin{split}
{\left. {\frac{{\partial {F_i}\left( {\alpha ,{\pmb \rho} } \right)}}{{\partial \alpha }}} \right|_{\alpha {\rm{ = }}1}} &> {\left. {\frac{{\partial {F_i}\left( {\alpha ,{\pmb \rho} } \right)}}{{\partial \alpha }}} \right|_{\alpha  = 1,{W_i} = 0}}\\
&= \frac{1}{{2{Y_i}{Z_i}}}\left\{ {{X_i} + {Y_i}{Z_i} + {Z_i}} \right.\\
&~~~~\left.- \sqrt {{{\left( {{X_i} + {Y_i}{Z_i} + {Z_i}} \right)}^2} - 4{X_i}{Y_i}Z_i}\right\}\\
&>0.
\end{split}
\end{equation}
Therefore, we can claim that $\alpha {{\mathcal B}}_i^{DF}\left( \pmb \rho  \right) > {{\mathcal B}}_i^{DF}\left( \alpha {\pmb \rho}  \right)$, which completes the proof.

\end{appendix}
\section{Acknowledgement}
The authors would like to thank the anonymous reviewers for
their valuable comments and suggestions, which improved the
quality of the paper. The authors also thank Dr. Peng Wang for his helpful discussion.

\ifCLASSOPTIONcaptionsoff
  \newpage
\fi

\bibliographystyle{IEEEtran}
\bibliography{References}

%\newpage
%

\end{document}